\pdfoutput=1
\documentclass{hkpaper}

\usepackage[english]{babel}

\usepackage{savethm}
\settheoremname{lemma}{Lemma}
\settheoremname{proposition}{Proposition}
\settheoremname{definition}{Definition}
\settheoremname{theorem}{Theorem}
\settheoremname{example}{Example}
\settheoremname{corollary}{Corollary}

\usepackage{microtype}

\usepackage{amsmath,amssymb,amsthm}
\makeatletter{}
\renewcommand{\phi}{\varphi}
\renewcommand{\epsilon}{\varepsilon}

\def\mean#1{[\![ \, #1 \, ]\!]}

\newcommand{\PMet}{\mathbf{PMet}}
\newcommand{\Set}{\mathbf{Set}}

\newcommand{\Id}{\mathrm{Id}}
\newcommand{\id}{\mathrm{id}}
\newcommand{\nonexpansiveTo}{\raisebox{-1.5pt}{\ensuremath{\,\mathop{\overset{1}{\to}}\,}}}
\newcommand{\cat}[1]{\mathbf{#1}}

\newcommand{\EM}{\ensuremath{\mathcal{EM}}}

\newcommand{\prealinf} {\ensuremath{[0,\infty]}}
\newcommand{\reals}{{[0,\top]}}
\newcommand{\N}{\mathbb{N}}
\newcommand{\set}[1]{\left\{#1\right\}}
\newcommand{\tuple}[1]{\left\langle#1\right\rangle}
\newcommand{\supp}{\ensuremath{\mathrm{supp}}}

\newcommand{\ev}{\ensuremath{ev}} 
\newcommand{\Kantorovich}[2]{\ensuremath{#2^{\,\uparrow {#1}}}} 

\newcommand{\LiftedFunctor}[1]{\ensuremath{\overline{#1}}}
\newcommand{\EvaluationFunctor}[1]{\ensuremath{\widetilde{#1}}}
\newcommand{\Wasserstein}[2]{\ensuremath{#2^{\,\downarrow {#1}}}} 

\newcommand{\LiftedMetric}[2]{\ensuremath{#2^{#1}}}
\newcommand{\Couplings}[1]{\Gamma_{#1}}
\newcommand{\Powerset}{\ensuremath{\mathcal{P}}}
\newcommand{\PowersetFinite}{\ensuremath{\Powerset_{\!f\!i\!n}}}
\newcommand{\Distributions}{\ensuremath{\mathcal{D}}}

\usepackage{tikz}
\usetikzlibrary{matrix}
\usetikzlibrary{arrows}
\usetikzlibrary{automata}
\usetikzlibrary{positioning}

\newenvironment{diagram}{\begin{tikzpicture}[->,auto,>=latex]}{\end{tikzpicture}}

\DeclareMathOperator*{\argmax}{arg\max}  

\theoremstyle{hkplain}

\newaliascnt{propositionApx}{theoremApx}

\aliascntresetthe{propositionApx}

\newaliascnt{lemmaApx}{theoremApx}
\newtheorem{lemmaApx}[lemmaApx]{Lemma}
\aliascntresetthe{lemmaApx}

\theoremstyle{hkdefinition}
\newaliascnt{definitionApx}{theoremApx}

\aliascntresetthe{definitionApx}

\newaliascnt{exampleApx}{theoremApx}

\aliascntresetthe{exampleApx}

\title{Towards Trace Metrics\\ via Functor Lifting\thanks{This is an extended version of \cite{BBKK15}. It consists of the material of the original paper and an appendix containing the proofs of the presented results.}}
\author[1]{Paolo Baldan}
\author[2]{Filippo Bonchi}
\author[3]{Henning Kerstan}
\author[3]{Barbara~König}
\affil[1]{Dipartimento di Matematica, Università di Padova, Italy\\\href{mailto:baldan@math.unipd.it}{{baldan@math.unipd.it}}}%
\affil[2]{CNRS, ENS Lyon, Université de Lyon, France\\\href{mailto:filippo.bonchi@ens-lyon.fr}{{filippo.bonchi@ens-lyon.fr}}}%
\affil[3]{Universität Duisburg-Essen, Germany\\\href{mailto:henning.kerstan@uni-due.de}{{henning.kerstan@uni-due.de}}, \href{mailto:barbara_koenig@uni-due.de}{{barbara\_koenig@uni-due.de}}}
\subjclass{F.3.1 Specifying and Verifying and Reasoning about Programs, D.2.4 Software/Program Verification}
\keywords{trace metric, monad lifting, pseudometric, coalgebra}
\newlist{well-behaved-axioms}{enumerate}{1}
\setlist[well-behaved-axioms,1]{%
	label=\textcolor{black!70}{\textsc{W\arabic*.}},%
	ref=(\arabic*)%
}

\newcommand{\iflongversion}[2]{#1} 
\newcommand{\appendixref}{\autoref{sec:proofs}}

\begin{document}
\renewcommand{\sectionautorefname}{Section}
\renewcommand{\subsectionautorefname}{Subsection}
\makeatletter{}
%
\maketitle 
\begin{abstract}
\noindent We investigate the possibility of deriving metric trace semantics in a coalgebraic framework. First, we generalize a technique for systematically lifting functors from the category $\Set$ of sets to the category $\PMet$ of pseudometric spaces, by identifying conditions under which also natural transformations, monads and distributive laws can be lifted. By exploiting some recent work on an abstract determinization, these results enable the derivation of trace metrics starting from coalgebras in $\Set$. More precisely, for a coalgebra on $\Set$ we determinize it, thus obtaining a coalgebra in the Eilenberg-Moore category of a monad. When the monad can be lifted to $\PMet$, we can equip the final coalgebra with a behavioral distance. The trace distance between two states of the original coalgebra is the distance between their images in the determinized coalgebra through  the unit of the monad. We show how our framework applies to nondeterministic automata and probabilistic automata.
\end{abstract}

\section{Introduction}
When considering the behavior of state-based system models embodying
quantitative information, such as probabilities, time or cost, the interest normally shifts from behavioral equivalences to behavioral distances. In fact, in a quantitative setting, it is often quite unnatural to ask that two systems exhibit exactly the same behavior, while it can be more reasonable to require that the distance between their behaviors is sufficiently small~(see, e.g.,~\cite{GJS90,DGJP04,vBW05,bblm:total-variation-markov,afs:linear-branching-metrics-quant,afs:linear-branching-metrics,flt:quantitative-spectrum}).

Coalgebras \cite{Rut00} are a well-established abstract framework where a canonical notion of behavioral equivalence can be uniformly derived. The behavior of a system is represented as a coalgebra, namely a map of the form $X \to H X$, where $X$ is a state space and $H$ is a functor that describes the type of computation performed. For instance nondeterministic automata can be seen as coalgebras $X \to 2 \times  \mathcal{P}(X)^A$: for any state we specify whether it is final or not, and the set of successors for any given input in $A$. Under suitable conditions a final coalgebra exists which can be seen as minimized version of the system, so that two states are deemed equivalent when they correspond to the same state in the final coalgebra.

In a recent paper~\cite{BBKK14} we faced the problem of devising a framework where, given a coalgebra for an endofunctor $H$ on $\Set$, one can  systematically derive pseudometrics which measure the behavioral distance of states. A first crucial step is the lifting of $H$ to a functor $\LiftedFunctor{H}$ on $\PMet$, the category of pseudometric spaces. In particular, we presented two different approaches which can be viewed as generalizations of the Kantorovich and Wasserstein pseudometrics for probability measures. One can prove that the final coalgebra in $\Set$ can be endowed with a metric, arising as a solution of a fixpoint equation, turning it into the final coalgebra for the lifting $\LiftedFunctor{H}$. Since any coalgebra $X \to HX$ can be seen as a coalgebra in $\PMet$ by endowing $X$ with the discrete metric, the unique mapping into the final coalgebra provides a behavioral distance~on~$X$.

The canonical notion of equivalence for coalgebras, in a sense, fully
captures the behavior of the system as expressed by the functor $H$. As such, it naturally corresponds to bisimulation equivalences already defined for various concrete formalisms. Sometimes  one is interested in coarser equivalences, ignoring some aspects of a computation, a notable example being trace equivalence where the computational effect which is ignored is branching.

In this paper, relying on recent work on an abstract determinization construction for coalgebras in~\cite{DBLP:journals/corr/abs-1302-1046,JSS12,JSS15}, we extend the above framework in order to systematically derive trace metrics.
The mentioned work starts from the observation that the distinction between the behavior to be observed and the computational effects that are intended to be hidden from the observer, is sometimes formally captured by splitting the functor $H$ characterizing system computations in two components, a functor $F$ for the observable behavior and a monad $T$ describing the computational effects, e.g., lifting $1 + -$, the powerset functor $\mathcal{P}$ or the distribution functor $\mathcal{D}$ provides partial, nondeterministic or probabilistic computations, respectively.
For instance, the functor for nondeterministic automata $2 \times  \Powerset(X)^A$ can be seen as the composition of the functor $F X = 2 \times X^A$, describing the transitions, with the powerset monad $T = \mathcal{P}$, capturing nondeterminism. Trace semantics can be derived by viewing a coalgebra $X \to 2 \times \mathcal{P}(X)^A$ as a coalgebra $\mathcal{P}(X) \to 2 \times \mathcal{P}(X)^A$, via a determinization construction.
Similarly probabilistic automata can be seen as coalgebras of the form
$X \to [0,1] \times \mathcal{D}(X)^A$, yielding coalgebras
$\mathcal{D}(X) \to [0,1] \times \mathcal{D}(X)^A$ via
determinization.

On this basis,~\cite{JSS15} develops a framework for deriving behavioral equivalences which only considers the visible behavior, ignoring the computational effects. The core idea consists in ``incorporating'' the effect of the monad also in the set of states $X$, which thus becomes $T X$, by means of a construction that can be seen as an abstract form of determinization. 
For functors of the shape $FT$, this can be done by lifting $F$ to a functor $\widehat{F}$ in $\EM(T)$, the Eilenberg-Moore category of $T$, using a distributive law between $F$ and $T$. In fact, the final $F$-coalgebra lifts to the final $\widehat{F}$-coalgebra in $\EM(T)$. The technique works, at the price of some complications, also for functors of the shape $TF$~\cite{JSS15}.

Here, we exploit the results in~\cite{JSS15} for
systematically deriving metric trace semantics for $\Set$-based
coalgebras. The situation is summarized in the diagram at the end of
\autoref{sec:generalized-powerset}.
As a first step, building on our technique for lifting functors from the category $\Set$ of sets to the category $\PMet$ of pseudometric spaces, we identify conditions under which also natural transformations, monads and distributive laws can be lifted.
In this way we obtain an adjunction between $\PMet$ and
$\EM(\LiftedFunctor{T})$, where $\LiftedFunctor{T}$ is the lifted
monad. Via the lifted distributive law we can transfer a functor
$\LiftedFunctor{F}\colon \PMet \to \PMet$ to an endofunctor
$\widehat{\LiftedFunctor{F}}$ on $\EM(\LiftedFunctor{T})$. By using the
trivial discrete distance, coalgebras of the form $TX\to FTX$ can now
live in $\EM(\LiftedFunctor{T})$ and can be equipped with a trace
distance via a map into the final coalgebra. This final coalgebra is
again obtained by lifting the final $\LiftedFunctor{F}$-coalgebra,
i.e.\ a coalgebra equipped with a behavioral distance, to
$\EM(\LiftedFunctor{T})$. 

The trace distance between two states of the original coalgebra can
then be defined as the distance between their images in the
determinized coalgebra through the unit of the monad. We illustrate
our framework by thoroughly discussing two running examples, namely
nondeterministic automata and probabilistic automata. We show that it
allows us to recover known or meaningful trace distances such as the
standard ultrametric on word languages for nondeterministic automata
or the total variation distance on distributions for probabilistic
automata.

The paper is structured as follows. In \autoref{sec:preliminaries} we
will introduce our notation and quickly recall the basics of our
lifting framework from~\cite{BBKK14}. Then, in
\autoref{sec:compositionality}, we tackle the question of
compositionality, i.e. we investigate whether based on liftings of two
functors we can obtain a lifting of the composed functor. The lifting
of natural transformations and monads is treated in
\autoref{sec:monadlifting}. Equipped with these tools, we show as main
result in \autoref{sec:tracemetrics} how to obtain trace pseudometrics
in the Eilenberg-Moore category of a lifted monad. We conclude our
paper with a discussion on related and future work
(\autoref{sec:conclusion}). Proofs can be found in \appendixref.

\section{Preliminaries}
\label{sec:preliminaries}
In this section we recap some basic notions and fix the corresponding
notation. We also briefly recall the results in~\cite{BBKK14} which
will be exploited in the paper.

We assume that the reader is familiar with the basic notions of category theory, especially with the definitions of functor, product, coproduct and weak pullbacks.

For a function $f\colon X \to Y$ and sets $A \subseteq X$, $B \subseteq Y$ we write $f[A] := \set{f(a) \mid a \in A}$ for the \emph{image} of $A$ and $f^{-1}[B]= \set{a \in A \mid f(x) \in B}$ for the \emph{preimage} of $B$. Finally, if $Y \subseteq [0,\infty]$ and $f,g\colon X \to Y$ are functions we write $f \leq g$ if $\forall x \in X:f(x)\leq g(x)$. 

A \emph{probability distribution} on a given set $X$ is a function $P\colon X \to [0,1]$ satisfying $\sum_{x \in X}P(x) = 1$. For any set $B \subseteq X$ we define $P(B) = \sum_{x \in B}P(x)$. The \emph{support} of $P$ is the set $\supp(P) := \set{x \in X\mid P(x) > 0}$.

Given a natural number $n \in \N$ and a family $(X_i)_{i = 1}^n$ of sets $X_i$ we denote the projections of the (cartesian) product of the $X_i$ by $\pi_i\colon \prod_{i=1}^n X_i \to X_i$. For a source $(f_i\colon X \to X_i)_{i = 1}^n$ we denote the unique mediating arrow to the product by $\langle f_1,\dots,f_n\rangle \colon X \to \prod_{i=1}^{n}X_i$. Similarly, given a family of arrows $(f_i\colon X_i \to Y_i)_{i = 1}^n$, we write $f_1 \times \dots \times f_n = \langle f_1 \circ \pi_1,\dots,f_n \circ \pi_n \rangle\colon \prod_{i=1}^n X_i \to \prod_{i=1}^n Y_i$.

For $\top \in (0,\infty]$ and a set $X$ we call any function $d\colon X^2 \to [0,\top]$ a ($\top$-)\emph{distance} on $X$ (for our examples we will use $\top = 1$ or $\top=\infty$). Whenever $d$ satisfies, for all $x,y,z \in X$, $d(x,x) = 0$ (reflexivity), $d(x,y) = d(y,x)$ (symmetry) and $d(x,y) \leq d(x,z) + d(z,y)$ (triangle inequality) we call it a \emph{pseudometric} and if it additionally satisfies $d(x,y) = 0 \implies x=y$ we call it a \emph{metric}. Given such a function $d$ on a set $X$, we say that $(X,d)$ is a pseudometric/metric space. By $d_e\colon \reals^2\to \reals$ we denote the ordinary Euclidean distance on $\reals$, i.e., $d_e(x,y) = |x-y|$ for $x,y\in\reals\setminus\{\infty\}$, and -- where appropriate -- $d_e(x,\infty) = \infty$ if $x\neq \infty$ and $d_e(\infty,\infty) = 0$. Addition is defined in the usual way, in particular $x + \infty = \infty$ for $x\in\prealinf$. We call a function $f\colon X \to Y$ between pseudometric spaces $(X,d_X)$ and $(Y,d_Y)$ \emph{nonexpansive} and write $f\colon (X,d_X) \nonexpansiveTo (Y,d_Y)$ if $d_Y \circ (f \times f) \leq d_X$. If equality holds we call $f$ an \emph{isometry}.

By choosing a fixed maximal element $\top$ in our definition of distances, we ensure that the set of pseudometrics over a fixed set with pointwise order is a complete lattice (since $\reals$ is) and we obtain a complete and cocomplete category of pseudometric spaces and nonexpansive functions, which we denote by $\PMet$. Given a functor $F$ on $\Set$, we aim at constructing a functor $\LiftedFunctor{F}$ on $\PMet$ which is a lifting of $F$ in the following sense.

\begin{definition}[Lifting]
\label{def:lifting}
Let $U\colon \PMet \to \Set$ be the forgetful functor which maps every pseudometric space to its underlying set. A functor $\LiftedFunctor{F}\colon\PMet\to\PMet$ is called a \emph{lifting} of a functor $F\colon \Set \to \Set$ if it satisfies $U\LiftedFunctor{F} = FU$.
\end{definition}

\noindent Similarly to predicate lifting of coalgebraic modal logic \cite{DBLP:journals/tcs/Schroder08}, lifting to $\PMet$ can be conveniently defined once a suitable (evaluation) function from
$F[0,\top]$ to $[0,\top]$ is fixed.

\begin{definition}[Evaluation Function \& Evaluation Functor]
\label{def:evfct} 
Let $F$ be an endofunctor on $\Set$. An \emph{evaluation function} for $F$ is a function $\ev_F\colon F[0,\top] \to [0,\top]$. Given such a function, we define the \emph{evaluation functor} to be the endofunctor $\EvaluationFunctor{F}$ on $\Set/[0,\top]$, the slice category\footnote{The slice category $\Set/[0,\top]$ has as objects all functions $g\colon X\to[0,\top]$ where $X$ is an arbitrary set. Given $g$ as before and $h\colon Y \to [0,\top]$, an arrow from $g$ to $h$ is a function $f\colon X \to Y$ satisfying $h \circ f = g$.} over $[0,\top]$, via $\EvaluationFunctor{F}(g) = \ev_F\circ Fg$ for all $g \in \Set/[0,\top]$. On arrows $\EvaluationFunctor{F}$ is defined as $F$. 
\end{definition}

\noindent A first lifting technique leads to what we called the
Kantorovich pseudometric, which is the smallest possible pseudometric
$d^F$ on $FX$ such that, for all nonexpansive functions
$f\colon (X,d) \nonexpansiveTo (\reals,d_e)$, also
$\EvaluationFunctor{F}f\colon (FX,d^F)\nonexpansiveTo (\reals,d_e)$ is
again nonexpansive.

\begin{definition}[Kantorovich Pseudometric \& Kantorovich Lifting]
\label{def:kantorovich}
Let $F\colon \Set \to \Set$ be a functor with an evaluation function $\ev_F$. For every pseudometric space $(X,d)$ the \emph{Kantorovich pseudometric} on $FX$ is the function $\Kantorovich{F}{d}\colon FX\times FX\to \reals$, where for all $t_1,t_2 \in FX$:
\begin{align*}
	\Kantorovich{F}{d}(t_1,t_2) := \sup \set{d_e(\EvaluationFunctor{F}f(t_1),\EvaluationFunctor{F}f(t_2)) \mid f\colon (X,d) \nonexpansiveTo (\reals,d_e)}\,.
\end{align*}
The \emph{Kantorovich lifting} of the functor $F$ is the functor $\LiftedFunctor{F}\colon \PMet\to\PMet$ defined as $\LiftedFunctor{F}(X,d) = (FX,\Kantorovich{F}{d})$ and $\LiftedFunctor{F}f = Ff$. 
\end{definition}

\noindent This definition is sound i.e. $\Kantorovich{F}{d}$ is
guaranteed to be a pseudometric so that we indeed obtain a lifting of
the functor. A dual way for obtaining a
pseudometric on $FX$ relies on ideas from probability and
transportation theory. It is based on the notion of couplings, which
can be understood as a generalization of joint probability measures.

\begin{definition}[Coupling]
\label{def:coupling}
Let $F \colon \Set \to \Set$ be a functor and $n \in \N$. Given a set $X$ and $t_i \in FX$ for $1 \leq i \leq n$ we call an element $t \in F(X^n)$ such that $F\pi_i(t) = t_i$ a \emph{coupling} of the $t_i$ (with respect to $F$). We write $\Couplings{F}(t_1, t_2, \dots, t_n)$ for the set of all these couplings. 
\end{definition}

\noindent Based on these couplings we are now able to define an alternative distance on $FX$.

\begin{definition}[Wasserstein Distance \& Wasserstein Lifting]
\label{def:wasserstein}
Let $F\colon \Set \to \Set$ be a functor with evaluation function $\ev_F$. For every pseudometric space $(X,d)$ the \emph{Wasserstein distance} on $FX$ is the function $\Wasserstein{F}{d} \colon FX \times FX \to \reals$ given by, for all $t_1,t_2 \in FX$,  
\begin{align*}
	\Wasserstein{F}{d}(t_1, t_2) := \inf\set{\EvaluationFunctor{F}d(t) \mid t \in \Couplings{F}(t_1,t_2)}\,.
\end{align*}
If $\Wasserstein{F}{d}$ is a pseudometric for all pseudometric spaces $(X,d)$, we define the \emph{Wasserstein lifting} of $F$ to be the functor $\LiftedFunctor{F}\colon \PMet\to\PMet$, $\LiftedFunctor{F}(X,d) = (FX,\Wasserstein{F}{d})$, $\LiftedFunctor{F}f = Ff$.
\end{definition}

\noindent The names Kantorovich and Wasserstein used for the liftings
derive from transportation theory \cite{v:optimal-transport}.
Indeed we obtain a transport problem if we instantiate $F$ with the
distribution functor $\mathcal{D}$ (see also
\autoref{exa:distributions} below). In order to measure the distance
between two probability distributions $s,t\colon X\to [0,1]$ it is
useful to think of the following analogy: assume that $X$ is a
collection of cities (with distance function $d$ between them) and
$s,t$ represent supply and demand (in units of mass).  The distance
between $s,t$ can be measured in two ways: the first is to set up an
optimal transportation plan with minimal costs (also called coupling)
to transport goods from cities with excess supply to cities with
excess demand. The cost of transport is determined by the product of
mass and distance.  In this way we obtain the Wasserstein distance.  A
different view is to imagine a logistics firm that is commissioned to
handle the transport.  It sets prices for each city and buys and sells
for this price at every location.  However, it has to ensure that the
price function (here, $f$) is nonexpansive, i.e., the difference of
prices between two cities is smaller than the distance of the cities,
otherwise it will not be worthwhile to outsource this task. This firm
will attempt to maximize its profit, which can be considered as the
Kantorovich distance of $s,t$. The Kantorovich-Rubinstein duality
informs us that these two views lead to the exactly same result

In \autoref{def:wasserstein} we are not guaranteed, in
general, that $\Wasserstein{F}{d}$ is a pseudometric. This
is the case if we require $F$ to preserve weak-pullbacks and
impose the following restrictions on the evaluation function.

\begin{definition}[Well-Behaved]\label{def:well-behaved}
Let $F$ be a functor with an evaluation function $\ev_F$. We call $\ev_F$ \emph{well-behaved} if it satisfies the following conditions:
\begin{well-behaved-axioms}
\item $\EvaluationFunctor{F}$ is monotone, i.e., for $f,g\colon X\to[0,\top]$ with $f\le g$, we have $\EvaluationFunctor{F}f\le\EvaluationFunctor{F}g$.\label{def:well-behaved:monotone}
\item For each $t\in F([0,\top]^2)$ it holds that $d_e(\ev_F(t_1), \ev_F(t_2)) \leq \EvaluationFunctor{F}d_e(t)$ for $t_i : = F\pi_i(t)$.\label{def:well-behaved:leq}
\item $\ev_F^{-1}[\set{0}] = Fi[F\{0\}]$ where $i \colon \set{0} \hookrightarrow[0,\top]$ is the inclusion map. \label{def:well-behaved:pullback}
\end{well-behaved-axioms}
\end{definition}

\noindent While condition W1 is quite natural, for W2 and W3 some explanations are in order. Condition~W2 ensures that $\EvaluationFunctor{F}\id_\reals = \ev_F\colon F\reals\to\reals$ is nonexpansive once $d_e$ is lifted to $F\reals$ (recall that for the Kantorovich lifting we require $\EvaluationFunctor{F}f$ to be nonexpansive for any nonexpansive $f$). Condition~W3 requires that exactly the elements of $F\{0\}$ are mapped to $0$ via $\ev_F$. This is necessary for reflexivity of the Wasserstein pseudometric. Indeed, with this definition at hand we were able to prove the desired result. 

\begin{proposition}[{\cite{BBKK14}}]\label{prop:wasserstein-pseudo}
If $F$ preserves weak pullbacks and $\ev_F$ is well-behaved, then $\Wasserstein{F}{d}$ is a pseudometric for any pseudometric space $(X,d)$.
\end{proposition}

\noindent From now on, whenever we use the Wasserstein lifting
$\Wasserstein{F}{d}$, we implicitly assume to be in the hypotheses of
\autoref{prop:wasserstein-pseudo}. 
It can be shown that, in general,
$\Kantorovich{F}{d} \leq \Wasserstein{F}{d}$. Whenever equality holds
we say that the functor and the evaluation function satisfy the
\emph{Kantorovich-Rubinstein duality}. This is helpful in many
situations (e.g., in \cite{vBW06} it allowed to reuse an efficient
linear programming algorithm to compute behavioral distance) but it is
usually difficult to
obtain.

We now recall two examples which will play an important role in this paper. First, we consider the following bounded variant of the powerset functor.

\begin{example}[{Finite Powerset}]
\label{exa:powersetfinite}
The finite powerset functor $\PowersetFinite$ assigns to each set $X$ the set $\PowersetFinite X = \set{S \subseteq X \mid |S| < \infty}$ and to each function $f\colon X \to Y$ the function $\PowersetFinite f\colon \PowersetFinite X \to \PowersetFinite Y$, $\PowersetFinite f(S) := f[S]$.
This functor preserves weak pullbacks and the evaluation function $\max\colon \PowersetFinite([0,\infty]) \to [0,\infty]$ with $\max \emptyset = 0$ is well-behaved. The Kantorovich-Rubinstein duality holds and the resulting distance is the Hausdorff pseudometric which, for any pseudometric space $(X,d)$ and any $X_1, X_2 \in \PowersetFinite X$, is defined as
\begin{align*}
	 d_H(X_1,X_2) = \max\left\{\max_{x_1\in X_1} \min_{x_2\in X_2} d(x_1,x_2), \max_{x_2\in X_2} \min_{x_1\in X_1} d(x_1,x_2) \right\}\,.
	 \end{align*}
\end{example}
\noindent Our second example is the following finite variant of the distribution functor.
\begin{savetheorem}[Finitely Supported Distributions]{example}{exa:distributions}
The probability distribution functor $\Distributions$ assigns to each
set $X$ the set $\Distributions X = \set{P\colon X \to [0,1] \
  \middle|\ |\supp(P)| < \infty, P(X) = 1}$ and
to each function $f\colon X \to Y$ the function $\Distributions
f\colon \Distributions X \to \Distributions Y$, $\Distributions
f(P)(y) = \sum_{x \in f^{-1}[\set{y}]} P(x) = P(f^{-1}[\set{y}])$. $\Distributions$ preserves weak pullbacks and the evaluation function $\ev_\Distributions \colon \Distributions[0,1] \to [0,1]$, $\ev_\Distributions(P) = \sum_{r \in [0,1]}r\cdot P(r)$ is well-behaved. For any pseudometric space $(X,d)$ we obtain the Wasserstein pseudometric which, for any $P_1, P_2 \in \Distributions X$, is defined as
\begin{align*}
	\Wasserstein{\Distributions }{d}(P_1,P_2)= \min \set{\sum_{x_1,x_2 \in X} d(x_1,x_2) \cdot P(x_1, x_2) \ \middle|\ P \in \Couplings{\Distributions }(P_1, P_2)}\,.
\end{align*}
The Kantorovich-Rubinstein duality \cite{v:optimal-transport} holds from classical results in transportation theory. 
\end{savetheorem}

\noindent While these two functors can be nicely lifted using the
theory developed so far, there are other functors that require a more
general treatment. For instance, consider the endofunctor $F = B
\times \_$ (left product with $B$) for some fixed $B$. Notice that for
$t_1,t_2 \in FX = B \times X$ with $t_i=(b_i,x_i)$ a coupling exists
iff $b_1 = b_2$. As a consequence, when $b_1 \neq b_2$, irrespectively of
the evaluation function we choose and of the distance between $x_1$
and $x_2$ in $(X, d)$, the lifted Wasserstein pseudometric will always
result in $ \Wasserstein{F}{d}(t_1,t_2) =\top$.  This can be
counterintuitive, e.g., taking $B=[0,1]$, $X \not=\emptyset$ and $t_1
= (0,x)$ and $t_2=(\epsilon,x)$ for a small $\epsilon > 0$ and an $x
\in X$.
The reason is that we think of $B = [0,1]$ as endowed with a non-discrete pseudometric, like e.g. the Euclidean metric $d_e$, plugged into the product after the lifting.
This intuition can be indeed formalized by considering the lifting of the product seen as a functor from $\Set \times \Set$ into $\Set$. 
More generally, it can be seen that the definitions and results introduced so far for endofunctors in $\Set$ straightforwardly extend to multifunctors on $\Set$, namely functors $F\colon \Set^n \to \Set$ on the product category $\Set^n$ for a natural number $n \in \N$. For ease of presentation we will not spell out the details here (they are spelled out in \cite{BBKK14}), but just provide an important example of a bifunctor (i.e. $n=2$).

\begin{savetheorem}[Product Bifunctor]{example}{exa:product-bifunctor}
The weak pullback preserving product bifunctor $F\colon \Set^2 \to \Set$ maps two sets $X_1, X_2$ to $F(X_1,X_2) = X_1\times X_2$ and two functions $f_i\colon X_i \to Y_i$ to the function $F(f_1,f_2) = f_1 \times f_2$. In this paper we will use the well-behaved evaluation functions $\ev_F\colon [0,1]^2 \to [0,1]$ presented in the table below. Therein we also list the pseudometric $\LiftedMetric{F}{(d_1,d_2)}\colon X_1 \times X_2 \to [0,\top]$ we obtain for pseudometric spaces $(X_1,d_1)$, $(X_2,d_2)$.
\begin{center}\begin{tabular}{c|c|c}
Parameters & $\ev_F(r_1,r_2)$ & $\LiftedMetric{F}{(d_1,d_2)}((x_1,x_2),(y_1,y_2))$ \\
\hline
$c_1,c_2 \in (0,1]$ & $\max \set{c_1r_1,c_2r_2}$ & $\max\set{c_1d_1(x_1,y_1), c_2d_2(x_2,y_2)}$\\
$c_1,c_2 \in (0,1], c_1+c_2 \leq 1$ & $c_1 x_1 + c_2 x_2$ &$c_1d_1(x_1,y_1)+c_2d_2(x_2,y_2)$
\end{tabular}\end{center}
For $c_1=c_2=1$, the first evaluation map yields exactly the categorical product in $\PMet$.
In both cases the Kantorovich-Rubinstein duality holds and the supremum [infimum] of the Kantorovich [Wasserstein] pseudometric is always a maximum [minimum]. 
\end{savetheorem} 

\section{Compositionality for the Wasserstein Lifting}
\label{sec:compositionality}

Our first step is to study compositionality of functor liftings, i.e.
we identify some sufficient conditions ensuring
$\LiftedFunctor{F}\,\LiftedFunctor{G} = \LiftedFunctor{FG}$. This
technical result will be often very useful since it allows us to
reason modularly and, consequently, to simplify the proofs needed in
the treatment of our examples. We will explicitly only consider the
Wasserstein approach which is the one employed in all the examples of
this paper.

Given evaluation functions $\ev_F$ and $\ev_G$, we can easily
construct an evaluation function for the composition $FG$ by defining
$\ev_{FG} := \EvaluationFunctor{F}\ev_G= \ev_F \circ F\ev_G$.  Our
first observation is that, whenever $F$ and $G$ preserve weak
pullbacks, well-behavedness is inherited.

\begin{savetheorem}[Well-Behavedness of Composed Evaluation Function]{proposition}{prop:comp:well-behaved}
Let $F$, $G$ be endofunctors on $\Set$ with evaluation functions $\ev_F$,  $\ev_G$. If both functors preserve weak pullbacks and both evaluation functions are well-behaved then also $\ev_{FG} = \ev_F \circ F\ev_G$ is well-behaved.
\end{savetheorem}
\noindent In the light of this result and the fact that $FG$ certainly preserves weak pullbacks if $F$ and $G$ do, we can safely use the Wasserstein lifting for $FG$. A sufficient criterion for compositionality is the existence of optimal couplings for $G$. 

\begin{savetheorem}[Compositionality]{proposition}{prop:comp}
Let $F,G$ be weak pullback preserving endofunctors on $\Set$ with well-behaved evaluation functions $\ev_F$, $\ev_G$ and let $(X,d)$ be a pseudometric space. Then $\Wasserstein{FG}{d} \geq \Wasserstein{F}{(\Wasserstein{G}{d})}$.
Moreover, if for all $t_1,t_2 \in GX$ there is an optimal $G$-coupling, i.e. $\gamma(t_1,t_2) \in \Gamma_G(t_1,t_2)$ such that $\Wasserstein{G}{d}(t_1,t_2) = \EvaluationFunctor{G}d(\gamma(t_1,t_2))$, then $\Wasserstein{FG}{d} = \Wasserstein{F}{(\Wasserstein{G}{d})}$.
\end{savetheorem}

\noindent This criterion will turn out to be very useful for our later
results. Nevertheless it provides just a sufficient condition for compositionality as the
next example shows.
\begin{savetheorem}{example}{exa:comp}
We consider the finite powerset functor $\PowersetFinite$ of \autoref{exa:powersetfinite} and the distribution functor $\Distributions$ of \autoref{exa:distributions} with their evaluation functions. Let $(X,d)$ be a pseudometric space.
\begin{enumerate}
\item We have $\Wasserstein{\Distributions\Distributions}{d} = \Wasserstein{\Distributions}{\left(\Wasserstein{\Distributions}{d}\right)}$, by~\autoref{prop:comp}, because optimal couplings always exist.
\item We have $\Wasserstein{\PowersetFinite\PowersetFinite}{d} = \Wasserstein{\PowersetFinite}{\left(\Wasserstein{\PowersetFinite}{d}\right)}$ although $\PowersetFinite$-couplings do not always exist.\label{exa:comp:3}
\end{enumerate}
\end{savetheorem}

\noindent Note that when we lift the functor $\PowersetFinite$ we do not have
couplings in the case when we determine the distance between an empty
set $\emptyset$ and a non-empty set $Y\subseteq X$, since there exists no subset of
$X\times X$ that projects to both.

Compositionality can be defined analogously for multifunctors. Again, we will not spell this out completely but we will use it to obtain the machine bifunctor. Before we can do that, we first need to define another endofunctor.

\begin{savetheorem}[Input Functor]{example}{exa:inputfunctor}
Let $A$ be a fixed finite set of inputs. The input functor $F=\_^A\colon \Set \to \Set$ maps a set $X$ to the exponential $X^A$ and a function $f\colon X \to Y$ to $f^A \colon X^A \to Y^A$, $f^A(g) = f \circ g$. This functor preserves weak pullbacks. The two evaluation functions listed below are well-behaved and yield the given Wasserstein pseudometric on $X^A$ for any pseudometric space $(X,d)$.
\begin{center}\begin{tabular}{c|c}
$\ev_F(s)$ & $\Wasserstein{F}{d}(s_1,s_2)$ \\
\hline
 $\max\iflongversion{\limits}{}_{a \in A} s(a)$ & $\max\iflongversion{\limits}{}_{a \in A}d\big(s_1(a), s_2(a)\big)$\\
$\sum\iflongversion{\limits}{}_{a \in A} s(a)$ &$\sum\iflongversion{\limits}{}_{a \in A}d\big(s_1(a), s_2(a)\big)$
\end{tabular}\end{center}
\end{savetheorem}

\noindent By composing this functor with the product bifunctor we obtain the machine bifunctor which we will use to obtain trace semantics.

\begin{savetheorem}[Machine Bifunctor]{example}{exa:machinebifunctor}
Let $A$ be a finite set of inputs, $I = \_^A$ the input functor of \autoref{exa:inputfunctor}, $\Id$ the identity endofunctor on $\Set$ and $P$ be the product bifunctor of \autoref{exa:product-bifunctor}. The \emph{machine bifunctor} is the composition $M := P \circ (\Id \times I)$ i.e. the bifunctor $M\colon \Set^2 \to \Set$ with $M(B,X) := B \times X^A$. Since for $\Id$ and $I$ there are unique (thus optimal) couplings we have compositionality. Depending on the choices of evaluation function for $P$ and $I$ (for $\Id$ we always take $\id_{[0,1]})$ we obtain the following well-behaved evaluation functions $\ev_M\colon [0,1] \times [0,1]^A \to [0,1]$.
\begin{center}\begin{tabular}{c|c|c|c}
Parameters & $\ev_P(r_1,r_2)$ & $\ev_I(s)$ & $\ev_M(o,s)$\\
\hline
$c_1,c_2 \in (0,1]$ & $\max \set{c_1r_1,c_2r_2}$ & $\max\iflongversion{\limits}{}_{a \in A}s(a)$& $\max\Big\{c_1o,c_2\max\iflongversion{\limits}{}_{a \in A}s(a)\Big\}$\\
\iflongversion{$\begin{matrix}c_1,c_2 \in (0,1], \\c_1+c_2 \leq 1\end{matrix}$}{$c_1,c_2 \in (0,1], c_1+c_2 \leq 1$} & $c_1 x_1 + c_2 x_2$ & $|A|^{-1}\sum\iflongversion{\limits}{}_{a \in A}s(a)$ &$c_1o + c_2|A|^{-1} \sum\iflongversion{\limits}{}_{a \in A}s(a)$.
\end{tabular}\end{center}
\noindent Let $(B,d_B)$, $(X,d)$ be pseudometric spaces. For any $t_1,t_2 \in M(B,X)$ with $t_i = (b_i,s_i) \in B \times X^A$ there is a unique and therefore necessarily optimal coupling $t := (b_1,b_2, \tuple{s_1,s_2})$. Depending on the evaluation function, we obtain for the first case%
\iflongversion%
{\begin{align*}\Wasserstein{M}{(d_B,d)}(t_1,t_2) = \max\set{c_1d_B(b_1,b_2), c_2 \cdot \max_{a \in A} d\big(s_1(a),s_2(a)\big)}\end{align*}}%
{$\Wasserstein{M}{(d_B,d)}(t_1,t_2) = \max\set{c_1d_B(b_1,b_2), c_2 \cdot \max_{a \in A} d(s_1(a),s_2(a))}$ }%
and for the second case 
\iflongversion{\begin{align*}\Wasserstein{M}{(d_B,d)}(t_1,t_2) ={c_1d_B(b_1,b_2) + {c_2}{|A|}^{-1} \sum_{a \in A} d\big(s_1(a),s_2(a)\big)}\,.\end{align*}}
{$\Wasserstein{M}{(d_B,d)}(t_1,t_2) ={c_1d_B(b_1,b_2) + {c_2}{|A|}^{-1} \sum_{a \in A} d(s_1(a),s_2(a))}$.}
\end{savetheorem}

\noindent Usually we will fix the first argument (the set of outputs) of the machine bifunctor and consider the obtained machine endofunctor $M_B := M(B,\_)$. However, for the same reasons as explained above for the product bifunctor, we need to consider it as bifunctor. One notable exception is the case where $B=2$, endowed with the discrete metric. Then we have the following result.

\begin{savetheorem}{example}{exa:m2}
Consider the machine endofunctor $M_2 := M(2,\_) = 2 \times \_^A$ with evaluation function $\ev_{M_2} \colon 2 \times [0,1]^A, (o,s) \mapsto c \cdot \ev_I(s)$ where $c \in (0,1]$ and $\ev_I$ is one of the evaluation functions for the input functor from \autoref{exa:inputfunctor}. If $d_2$ is the discrete metric on $2$ and $c = c_2$ (where $c_2$ is the parameter for the evaluation function of the machine bifunctor as in \autoref{exa:machinebifunctor}) then the pseudometric obtained via the bifunctor lifting coincides with the one obtained by endofunctor lifting i.e. for all pseudometric spaces $(X,d)$ we have $\Wasserstein{M}{(d_2,d)} = \Wasserstein{M_2}{d}$. Moreover, although couplings for $M_2$ do not always exist we have $\Wasserstein{\PowersetFinite M_2}{d} = \Wasserstein{\PowersetFinite}{\left(\Wasserstein{M_2}{d}\right)}$.
\end{savetheorem}

\section{Lifting of Natural Transformations and Monads}
\label{sec:monadlifting}
Recall that a monad on an arbitrary category $\cat{C}$ is a triple $(T, \eta, \mu)$ where $T \colon \cat{C} \to \cat{C}$ is an endofunctor and $\eta \colon \Id \Rightarrow T$, $\mu\colon T^2 \Rightarrow T$ are natural transformations called \emph{unit} ($\eta$) and \emph{multiplication} ($\mu$) such that the two diagrams below commute.
\begin{center}\begin{diagram}
\matrix[matrix of math nodes, column sep=1.6cm, row sep=.5cm] (m){
	T & T^2 & T && T^3 & T^2\\
	&T&&&T^2&T\\
};

\path (m-1-1) edge node[above] {$\eta T$} (m-1-2);
\path (m-1-3) edge node[above] {$T\eta$} (m-1-2);
\path (m-1-2) edge node[left] {$\mu$} (m-2-2);
\draw[-] (m-1-1) edge[double,double distance=2pt] (m-2-2);
\draw[-] (m-1-3) edge[double,double distance=2pt] (m-2-2);

\path (m-1-5) edge node[above] {$\mu T$} (m-1-6);
\path (m-1-5) edge node[left] {$T\mu$} (m-2-5);
\path (m-1-6) edge node[right] {$\mu$} (m-2-6);
\path (m-2-5) edge node[above] {$\mu$} (m-2-6);
\end{diagram}\end{center}
If we have a monad on $\Set$, we can of course use our framework to lift the endofunctor $T$ to a functor $\LiftedFunctor{T}$ on pseudometric spaces. A natural question that arises is, whether we also obtain a monad on pseudometric spaces, i.e. if the components of the unit and the multiplication are nonexpansive with respect to the lifted pseudometrics. In order to answer this question, we first take a closer look at sufficient conditions for lifting natural transformations.

\begin{savetheorem}[Lifting of a Natural Transformation]{proposition}{prop:nt-lifting}
Let $F$, $G$ be endofunctors on $\Set$ with evaluation functions
$\ev_F$, $\ev_G$ and $\lambda\colon F \Rightarrow G$ be a natural
transformation. Then the following holds for all pseudometric spaces
$(X,d)$. For the Kantorovich lifting:
\begin{enumerate}
\item If $\ev_{G} \circ \lambda_{[0,\top]} \leq \ev_{F}$ then $\Kantorovich{G}{d} \circ (\lambda_X \times \lambda_X) \leq \Kantorovich{F}{d}$, i.e. $\lambda_X$ is nonexpansive.\label{prop:nt-lifting:1}
\item If $\ev_{G} \circ \lambda_{[0,\top]} = \ev_{F}$ then $\Kantorovich{G}{d} \circ (\lambda_X \times \lambda_X) = \Kantorovich{F}{d}$, i.e. $\lambda_X$ is an isometry.\label{prop:nt-lifting:2}
\end{enumerate}
\noindent
while for the Wasserstein lifting 
\begin{enumerate}
\setcounter{enumi}{2}
\item If $\ev_{G} \circ \lambda_{[0,\top]} \leq \ev_{F}$ then $\Wasserstein{G}{d} \circ (\lambda_X \times \lambda_X) \leq \Wasserstein{F}{d}$, i.e. $\lambda_X$ is nonexpansive.\label{prop:nt-lifting:3}
		\item If $\ev_{G} \circ \lambda_{[0,\top]} = \ev_{F}$ and the Kantorovich Rubinstein duality holds for $F$, i.e. $\Kantorovich{F}{d} = \Wasserstein{F}{d}$, then $\Wasserstein{G}{d} \circ (\lambda_X \times \lambda_X) = \Wasserstein{F}{d}$, i.e. $\lambda_X$ is an isometry.\label{prop:nt-lifting:4}
\end{enumerate}
\end{savetheorem}

\noindent In the rest of the paper we will call a natural transformation $\lambda$ nonexpansive [an isometry] if (and only if) each of its components are nonexpansive [isometries] and write $\overline{\lambda}$ for the resulting natural transformation from $\LiftedFunctor{F}$ to $\LiftedFunctor{G}$. Instead of checking nonexpansiveness separately for each component of a natural transformation, we can just check the above (in-)equalities involving the two evaluation functions. 

By applying these conditions on the unit and multiplication of a given monad, we can now provide sufficient criteria for a monad lifting.

\begin{savetheorem}[Lifting of a Monad]{corollary}{cor:monad-lifting}
Let $(T, \eta, \mu)$ be a $\Set$-monad and $\ev_T$ an evaluation function for $T$. Then the following holds.
\begin{enumerate}
\item If $\ev_T \circ \eta_{[0,\top]} \leq \id_{[0,\top]}$ then $\eta$ is nonexpansive for both liftings. Hence we obtain the unit $\overline{\eta}\colon \LiftedFunctor{\Id}\Rightarrow \LiftedFunctor{T}$ in $\PMet$. 
\item If $\ev_T \circ \eta_{[0,\top]} = \id_{[0,\top]}$ then $\eta$ is an isometry for both liftings. 
\item Let $\LiftedMetric{T}{d} \in \{\Kantorovich{T}{d}, \Wasserstein{T}{d}\}$. If $\ev_T \circ \mu_{[0,\top]} \leq \ev_T \circ T\ev_T$ and compositionality holds for $TT$, i.e. $\LiftedMetric{T}{(\LiftedMetric{T}{d})} = \LiftedMetric{TT}{d}$, then $\mu$ is nonexpansive, i.e. $\LiftedMetric{T}{d} \circ (\mu_X \times \mu_X) \leq \LiftedMetric{T}{(\LiftedMetric{T}{d})}$. This yields the multiplication $\overline{\mu}\colon \LiftedFunctor{T}\,\LiftedFunctor{T} \Rightarrow \LiftedFunctor{T}$ in $\PMet$.
\end{enumerate}
\end{savetheorem}

\noindent We conclude this section with two examples of liftable monads.

\begin{example}[Finite Powerset Monad]\label{exa:monad:powfin}
The finite powerset functor $\PowersetFinite$ of
\autoref{exa:powersetfinite} can be seen as a monad, with unit $\eta$
consisting of the functions $\eta_X\colon X \to \PowersetFinite X$,
$\eta_X(x) = \set{x}$ and multiplication given by $\mu_X \colon
\PowersetFinite \PowersetFinite X \to \PowersetFinite X$, $\mu_X(S) =
\cup S$. We check if our conditions for the Wasserstein lifting are satisfied. Given $r \in [0,\infty]$ we have $\ev_T \circ \eta_{[0,\infty]}(r) = \max\set{r} = r$ and for $\mathcal{S} \in \PowersetFinite(\PowersetFinite[0,\top])$ we have $\ev_T \circ \mu_{[0,1]}(\mathcal{S}) = \max \cup \mathcal{S} = \max \cup_{S \in \mathcal{S}} S$ and $\ev_T \circ T\ev_T (\mathcal{S}) = \max\left(\ev_T[\mathcal{S}]\right) = \max\set{\max S \mid S \in \mathcal{S}}$ and thus both values coincide. Moreover, we recall from \autoref{exa:comp}.\ref{exa:comp:3} that we have compositionality for $\PowersetFinite\PowersetFinite$. Therefore, by \autoref{cor:monad-lifting} $\eta$ is an isometry and $\mu$ nonexpansive.
\end{example}

\begin{savetheorem}[Distribution Monad]{example}{exa:monad:distr}
The probability distribution functor $\mathcal{D}$ of \autoref{exa:distributions} can be seen as a monad: the unit $\eta$ consists of the functions $\eta_X\colon X \to \Distributions X$, $\eta_X(x) = \delta_{x}^X$ where $\delta_x^X$ is the Dirac distribution and the multiplication is given by $\mu_X \colon \Distributions\Distributions X \to \Distributions X$, $\mu_X(P) = \lambda x.\sum_{q \in \Distributions X}P(q) \cdot q(x)$. We consider its Wasserstein lifting. Since $[0,1] = \Distributions 2$ we can see that $\ev_\Distributions = \mu_2$. Using this and the monad laws we have $\ev_\Distributions \circ \eta_{[0,1]} = \mu_2 \circ \eta_{\Distributions 2} = \id_{\Distributions X} = \id_{[0,1]}$ and also $\ev_\Distributions \circ \mu_{[0,1]} = \mu_2 \circ \mu_{\Distributions 2} = \mu_2 \circ \Distributions \mu_2 = \ev_\Distributions \circ \Distributions\ev_\Distributions$. Moreover, since we always have optimal couplings, we have compositionality for $\Distributions\Distributions$ by \autoref{prop:comp}. Thus by \autoref{cor:monad-lifting} $\eta$ is an isometry and $\mu$ nonexpansive.
\end{savetheorem}

\section{Trace Metrics in Eilenberg-Moore} 
\label{sec:tracemetrics}
As mentioned in the introduction, trace semantics can be characterized by means of coalgebras either over Kleisli \cite{DBLP:journals/tcs/PowerT99,HJS07} or over Eilenberg-Moore \cite{DBLP:journals/corr/abs-1302-1046,JSS15} categories. We focus on the latter approach. We first recall the basic notions of Eilenberg-Moore algebras and distributive laws, and discuss how the results in the paper can be used to ``lift'' the associated determinization construction. This is then applied to derive trace metrics for nondeterministic automata and probabilistic automata, by relying on suitable liftings of the machine functor.

\subsection{Generalized Powerset Construction}
\label{sec:generalized-powerset}
An \emph{Eilenberg-Moore algebra} for a monad $(T,\eta,\mu)$ is a $\cat{C}$-arrow $a\colon TA \to A$ making the left and middle diagram below commute. Given two such algebras $a\colon TA \to A$ and $b\colon TB \to B$, a morphism from $a$ to $b$ is a $\cat{C}$ arrow $f\colon A \to B$ making the right diagram below commute.
 \begin{center}\begin{diagram}
\matrix[matrix of math nodes, column sep=1.2cm, row sep=.5cm] (m){
	A & TA && T^2A & TA && TA & TB\\
	&A 	&& TA & A && A & B\\
};
\draw[-] (m-1-1) edge[double,double distance=2pt] (m-2-2);
\draw (m-1-1) edge node[above]{$\eta_A$}(m-1-2);
\draw (m-1-2) edge node[right]{$a$}(m-2-2);

\draw (m-1-4) edge node[above]{$\mu_A$}(m-1-5);
\draw (m-1-5) edge node[right]{$a$}(m-2-5);
\draw (m-1-4) edge node[left]{$Tc$}(m-2-4);
\draw (m-2-4) edge node[below]{$a$}(m-2-5);

\draw (m-1-7) edge node[above]{$Tf$}(m-1-8);
\draw (m-1-8) edge node[right]{$b$}(m-2-8);
\draw (m-1-7) edge node[left]{$a$}(m-2-7);
\draw (m-2-7) edge node[below]{$f$}(m-2-8);
\end{diagram}\end{center}
Eilenberg-Moore algebras and their morphisms form a category denoted by $\EM(T)$.
A functor $\widehat{F}\colon\EM(T)\to\EM(T)$ is called a \emph{lifting} of $F\colon \cat{C} \to \cat{C}$ to $\EM(T)$ if $U^T\widehat{F} = FU^T$,
with $U^T\colon \EM(T)\to \cat{C}$ the forgetful functor. A natural transformation $\lambda\colon TF \Rightarrow FT$ is an \emph{$\EM$-law} (also called \emph{distributive law}) if it satisfies:
 \begin{center}\begin{diagram}
\matrix[matrix of math nodes, column sep=1.6cm, row sep=.5cm] (m){
	F & F 	&& T^2F & TFT & FT^2\\
	TF &FT 	&& TF & & FT\\
};
\draw[-] (m-1-1) edge[double,double distance=2pt] (m-1-2);
\draw (m-1-1) edge node[left]{$\eta F$}(m-2-1);
\draw (m-1-2) edge node[right]{$F\eta$}(m-2-2);
\draw (m-2-1) edge node[below]{$\lambda$}(m-2-2);

\draw (m-1-4) edge node[above]{$T\lambda$}(m-1-5);
\draw (m-1-5) edge node[above]{$\lambda T$}(m-1-6);
\draw (m-1-4) edge node[left]{$\mu F$}(m-2-4);
\draw (m-2-4) edge node[below]{$\lambda$}(m-2-6);
\draw (m-1-6) edge node[right]{$F\mu$}(m-2-6);
\end{diagram}\end{center}
Liftings and $\EM$-laws are related by the following folklore result (see e.g.~\cite{JSS12}).
\begin{proposition}
There is a bijective correspondence between $\EM$-laws and liftings to $\EM$-categories.
\end{proposition}

\noindent $\EM$-laws and liftings are crucial to characterize trace semantics
via coalgebras. Given a coalgebra $c\colon X\to FTX$, for a functor
$F$ and a monad $(T,\eta,\mu)$ such that there is a distributive law $\lambda\colon TF \Rightarrow FT$, one can build an $F$-coalgebra as 
\begin{center}\begin{diagram}
\matrix[matrix of math nodes, column sep=1.6cm, row sep=.5cm] (m){
	c^\sharp := \Big(TX & TFTX & FTTX & FTX\big)\\
};
\draw (m-1-1) edge node[above]{$Tc$}(m-1-2);
\draw (m-1-2) edge node[above]{$\lambda_{TX}$}(m-1-3);
\draw (m-1-3) edge node[above]{$F\mu_X$}(m-1-4);
\end{diagram}\end{center}
If there exists a final $F$-coalgebra $\omega\colon \Omega \to F\Omega$, one can define a semantic map for the $FT$-coalgebra $c$ into $\Omega$. First let $\mean{-}\colon TX \to \Omega$ be the unique coalgebra morphism from $c^{\sharp}$. Then take the map $\mean{-}\circ \eta\colon X \to \Omega$.

 \begin{center}\begin{diagram}
\matrix[matrix of math nodes, column sep=1cm, row sep=1cm] (m){
	X & TX 	 && \Omega \\
	FTX & 	&& T \Omega \\
};
\draw (m-1-1) edge node[above]{$\eta$}(m-1-2);
\draw (m-1-2) edge node[above]{$\mean{-}$}(m-1-4);
\draw (m-1-1) edge node[left]{$c$} (m-2-1);
\draw (m-1-4) edge node[right]{$\omega$} (m-2-4);
\draw (m-2-1) edge node[below]{$F\mean{-}$}(m-2-4);
\draw (m-1-2) edge node[right]{$c^{\sharp}$}(m-2-1);

\end{diagram}\end{center}

\noindent One can readily check that $c^{\sharp}$ is an algebra map from the $T$-algebra $\mu_X$ to $\widehat{F}\mu_X$, namely it is an 
$\widehat{F}$-coalgebra or,
equivalently, a $\lambda$-\emph{bialgebra} \cite{DBLP:conf/lics/TuriP97,DBLP:journals/tcs/Klin11}. Similarly for $\omega$,
$\Omega$ carries a $T$-algebra structure obtained by finality and
hence the final $F$-coalgebra $\omega$ can be lifted in order to
obtain the final $\widehat{F}$-coalgebra (see \cite[Prop. 4]{JSS12}).

This result holds for arbitrary categories and, in particular, we can reuse it for our setting: we only need an $\EM$-law on $\PMet$. Note that
\autoref{prop:nt-lifting} not only provides sufficient conditions for monad
liftings but also can be exploited to lift $\EM$-laws. Indeed the additional commutativity requirements
for $\EM$-laws trivially hold when all components are nonexpansive.

\begin{savetheorem}[Lifting of an $\EM$-law]{corollary}{cor:lifting-distr-law}
Let $F,G$ be weak pullback preserving endo\-functors on $\Set$ with well-behaved evaluation functions $\ev_F$, $\ev_G$ and $\lambda\colon FG \Rightarrow GF$ be an $\EM$-law. If the evaluation functions satisfy $\ev_G \circ G\ev_F \circ \lambda_{[0,\top]} \leq \ev_F \circ F\ev_G$ and compositionality holds for $FG$, then $\lambda$ is nonexpansive and hence $\overline{\lambda}\colon \LiftedFunctor{F}\,\LiftedFunctor{G} \Rightarrow \LiftedFunctor{G}\,\LiftedFunctor{F}$ is also an $\EM$-law.
\end{savetheorem}

\noindent We will now consider $\EM$-laws for nondeterministic
and probabilistic automata. In the first case, $T$ is the powerset
monad $\PowersetFinite$ and $F$ is the machine functor $M_2 = 2\times
\_^A$, while in the second case $T$ is the distribution monad
$\mathcal{D}$ and $F$ is the machine functor $M_{[0,1]} = [0,1]\times
\_^A$. Note however that while in the first case
\autoref{cor:lifting-distr-law} is directly applicable, this is not
true in the second case, since we need to deal with multifunctors.

\begin{savetheorem}[\EM-law for Nondeterministic Automata]{example}{exa:EMlaw:powfin}
Let $(\PowersetFinite,\eta,\mu)$ be the finite powerset monad from \autoref{exa:monad:powfin}. The $\EM$-law $\lambda \colon \PowersetFinite(2 \times \_^A) \Rightarrow 2 \times \PowersetFinite(\_)^A$ is defined, for any set $X$, as
\begin{align*}
	\lambda_X(S) = \big(o,\lambda a \in A.\set{s'(a) \mid (o',s') \in S}\big), \ \ \text{where} \ \ o = \begin{cases}1 & \exists s' \in X^A.(1,s') \in S\\0&\text{else}\end{cases} \,.
\end{align*}
This is exactly the one exploited for the standard powerset construction from automata theory~\cite{DBLP:journals/corr/abs-1302-1046}. Indeed, for a nondeterministic automaton $c \colon X \to 2 \times \PowersetFinite(X)^A$, the map $\mean{-} \circ \eta_X$ assigns to each state its accepted language.
\autoref{cor:lifting-distr-law} ensures that it is nonexpansive (see \appendixref\ for a detailed proof). 
\end{savetheorem}

\begin{savetheorem}[\EM-law for Probabilistic Automata]{example}{exa:EMlaw:distr}
Let $(\Distributions,\eta,\mu)$ be the distribution monad from \autoref{exa:monad:distr} and $M$ be the machine bifunctor from \autoref{exa:machinebifunctor}. There is a known~\cite{DBLP:journals/corr/abs-1302-1046} $\EM$-law $\lambda \colon \Distributions([0,1] \times \_^A) \Rightarrow [0,1] \times \Distributions^A$ given by the assignment
\begin{align*}
	\lambda_X(P) = \left(\sum_{r \in [0,1]} r \cdot P(r,X^A), \lambda a \in A.\lambda x \in X. \sum_{{s \in X^A,\, s(a)=x}}\!\!\!\!\!\!\!\!\!\!P([0,1],s)\right)
\end{align*}
Also this $\EM$-law is nonexpansive, as shown in \appendixref. 
\end{savetheorem}

\noindent Any $FT$-coalgebra $c\colon X \to FTX$ can always be regarded as an $\LiftedFunctor{F}\,\LiftedFunctor{T}$-coalgebra by
equipping $X$ with the discrete metric assigning $\top$ to non equal states (in this way, $c$ is trivially nonexpansive).
The consequence of the nonexpansiveness of the $\EM$-laws $\lambda$ is
the following: the ``generalized determinization'' procedure for nondeterministic and probabilistic automata
can now be lifted to pass from $\LiftedFunctor{F}\,\LiftedFunctor{T}$-coalgebras to 
$\widehat{\LiftedFunctor{F}}$-coalgebras in $\EM(\LiftedFunctor{T})$ by using the upper adjunction in the diagram below (analogously to
\cite{JSS12, JSS15}).

\begin{center}\begin{diagram}
	\matrix[matrix of math nodes, column sep=2cm, row sep=2cm] (m){
		\PMet & \EM(\LiftedFunctor{T})\\
		\Set & \EM(T)\\
	};
	
	\path (m-1-1) edge node[left]{$U$} (m-2-1);
	\path (m-1-2) edge node[right]{$V$} (m-2-2);
	\path (m-1-1) edge[bend left=20] node[above]{$L^{\LiftedFunctor{T}}$} (m-1-2);
	\path (m-1-2) edge[bend left=20] node[below]{$U^{\LiftedFunctor{T}}$} (m-1-1);
	
	\path (m-2-1) edge[bend left=20] node[above]{$L^T$} (m-2-2);
	\path (m-2-2) edge[bend left=20] node[below]{$U^T$} (m-2-1);
        \path (m-1-1) edge[loop left] node[left]{$\LiftedFunctor{F}$}
        (m-1-1);
        \path (m-2-1) edge[loop left,looseness=20] node[left]{$F$} 
        (m-2-1);
        \path (m-1-2) edge[loop right]
        node[right]{$\widehat{\LiftedFunctor{F}}$} 
        (m-1-2);
        \path (m-2-2) edge[loop right] node[right]{$\widehat{F}$} 
        (m-2-2);
\end{diagram}\end{center}

\noindent Since we can also lift the final $\LiftedFunctor{F}$-coalgebra to
$\EM(\LiftedFunctor{T})$, we can use it to define trace distance.
This procedure is detailed in the next section.

\subsection{Final Coalgebra for the Lifted Machine Functor}
If we fix the first component of the machine bifunctor $M$ on $\Set$
we obtain an endofunctor $M_B\colon \Set \to \Set$, $M_B(X) =B \times
\_^A$. It is known~\cite{AM86} that the final coalgebra for this functor is
$\kappa \colon B^{A^*} \to B \times (B^{A^*})^A$ with $\kappa(t) =
(t(\epsilon),\lambda a \in A.\lambda w \in A^*.t(aw))$. We employ an
analogous construction with our lifted machine bifunctor
$\LiftedFunctor{M}$ on $\PMet$, i.e. we fix a pseudometric space $(B,
d_B)$ of outputs and consider coalgebras of the functor
$\LiftedFunctor{M}_{(B,d_B)} := \LiftedFunctor{M}((B, d_B), \_)$. To
obtain the final coalgebra for this functor in $\PMet$, we use the
following result from \cite{BBKK14}.

\begin{proposition}[{\cite[Thm. 6.1]{BBKK14}}]
\label{prop:final-coalgebra}
Let $\LiftedFunctor{F}\colon\PMet\to\PMet$ be a lifting of a functor $F\colon \Set\to \Set$ which has a final coalgebra $\kappa\colon\Omega\to F\Omega$. For every ordinal $i$ we construct a pseudometric $d_i\colon\Omega\times\Omega\to\reals$ as follows: $d_0 := 0$ is the zero pseudometric, $d_{i+1} := \LiftedMetric{F}{d_i}\circ(\kappa\times\kappa)$ for all ordinals $i$ and $d_j = \sup_{i<j} d_i$ for all limit ordinals $j$.
This sequence converges for some ordinal $\theta$, i.e $d_\theta = \LiftedMetric{F}{d_\theta}\circ(\kappa\times\kappa)$. Moreover $\kappa\colon (\Omega,d_\theta) \nonexpansiveTo (F\Omega,\LiftedMetric{F}{d_\theta})$ is the final $\LiftedFunctor{F}$-coalgebra.
\end{proposition}

\noindent It is hence enough to do fixed-point iteration for the functor $F$ on
the determinized state set $TX$ in order to obtain trace distance. The
lifted monad is ignored at this stage, but its lifting is of course
necessary to establish the Eilenberg-Moore category and its
adjunction.

We now consider our two example cases, where in both cases $F$ is the
machine functor $M_B$ (for two different choices of $B$):

\begin{savetheorem}[Final Coalgebra Pseudometric]{example}{exa:final-coalgebra}
Let $M$ be the machine bifunctor. 
\begin{enumerate}
\item We start with nondeterministic automata where the output set is $B = 2$ and we
  use the discrete metric $d_2$ as distance on $2$ as in \autoref{exa:m2}. As maximal distance we take $\top=1$ and as evaluation function we use $\ev_M(o,s) = c \cdot \max_{a \in A} s(a)$ for $0<c <1$. 
  
For any pseudometric $d$ on $2^{A^*}$ -- the carrier of the final $M_2$-coalgebra -- we know that for elements $(o_1,s_1), (o_2, s_2) \in 2 \times (2^{A^*})^A$ we have the Wasserstein pseudometric $\Wasserstein{F}{d}\big((o_1,s_1), (o_2, s_2)\big) = \max\set{d_2(o_1,o_2), c \cdot \max_{a \in A}d\big(s_1(a), s_2(a)\big)}$. Thus the fixed-point equation from \autoref{prop:final-coalgebra} is, for $L_1,L_2 \in 2^{A^*}$,
\begin{align*}
	d(L_1,L_2) = \max\set{d_2\big(L_1(\epsilon), L_2(\epsilon)\big), c \cdot \max_{a \in A} d\big(\lambda w.L_1(aw), \lambda w.L_2(aw)\big)}
\end{align*}
Now because $d_2$ is the discrete metric with $d_2(0,1) = 1$ we see that $d_{2^{A^*}}$ as defined below is indeed the least fixed-point of this equation and thus $(2^{A^*}, d_{2^{A^*}})$ is the carrier of the final $\LiftedFunctor{M_2}$-coalgebra.
  \begin{align*}
    d_{2^{A^*}} \colon 2^{A^*} \times 2^{A^*} \to [0,1], \quad
    d_{2^{A^*}}(L_1,L_2) = c^{\inf\set{n \in N \mid \exists w \in
        A^n.L_1(w) \not = L_2(w)}}\,.
  \end{align*} 
  
  A determinized coalgebra has as carrier set sets of states
  $\mathcal{P}(X)$. Each of these sets is mapped to the language that
  it accepts and the distance between two languages $L_1,L_2\colon
  A^*\to 2$ can be determined by looking for a word $w$ of minimal length which
  is contained in one and not in the other. Then, the distance is
  computed as $c^{|w|}$. This corresponds to the standard ultrametric
  on words.
  
\item Next we consider probabilistic automata where $B = [0,1]$ equipped with 
 the standard Euclidean metric $d_e$. 

 Furthermore the remaining parameters are set as follows: let $\top=1$
 and the evaluation function is
 $\ev_M(o,s) = c_1o + {c_2}{|A|}^{-1} \sum_{a \in A} s(a)$ for
 $c_1,c_2 \in (0,1)$ such that $c_1 + c_2 \leq 1$ as in
 \autoref{exa:machinebifunctor}. This time, the machine functor must
 be lifted as a bifunctor in order to obtain the appropriate distance
 (cf.\ the discussion before \autoref{exa:product-bifunctor}).
  
  For any pseudometric $d$ on $[0,1]^{A^*}$ we know that for $(r_1,s_1), (r_2, s_2) \in [0,1] \times ([0,1]^{A^*})^A$ we have $\Wasserstein{F}{d}((r_1,s_1), (r_2, s_2)) = c_1|r_1-r_2| +  \frac{c_2}{|A|}\cdot\sum_{a \in A}d(s_1(a), s_2(a))$. Thus the fixed-point equation from \autoref{prop:final-coalgebra} is, for $p_1,p_2 \in [0,1]^{A^*}$:
\begin{align*}
	d(p_1,p_2) = c_1|p_1(\epsilon)- p_2(\epsilon)| + \frac{c_2}{|A|} \cdot \sum_{a \in A} d\Big(\lambda w.p_1(aw), \lambda w.p_2(aw)\Big)
\end{align*}
It is again easy to see that $d_{[0,1]^{A^*}} \colon [0,1]^{A^*} \times [0,1]^{A^*}
  \to [0,1]$ as presented below is the least fixed-point of this equation and therefore $([0,1]^{A^*}, d_{[0,1]^{A^*}})$ the carrier of the final $\LiftedFunctor{M}_{([0,1],d_e)}$-coalgebra.
  \begin{align*}
    d_{[0,1]^{A^*}}(p_1,p_2) = c_1 \cdot \sum_{w \in A^*} \left(\frac{c_2}{|A|}\right)^{|w|} \left|p_1(w) - p_2(w)\right|\,.
  \end{align*}
  Here, a determinized coalgebra has as carrier distributions on
  states $\mathcal{D}(X)$. Each such distribution is mapped to a
  function $p\colon A^*\to [0,1]$ assigning numerical values to words.
  Then the distance, which can be thought of as a form of total
  variation distance with discount, is computed by the above formula.

  If instead of working in the interval $[0,1]$ we use $[0,\infty]$
  with $\top=\infty$, we can drop the conditions $c_1,c_2 < 1$ and
  $c_1+c_2\le 1$. In this case we may set $c_2 := |A|$ and $c_1 :=
  1/2$ and then the above distance is equal to the total variation
  distance, i.e.,
    \begin{align*}
    d_{[0,\infty]^{A^*}}(p_1,p_2) = \frac{1}{2} \cdot \sum_{w \in A^*} 
    \left|p_1(w) - p_2(w)\right| \,.
  \end{align*}

\end{enumerate}
\end{savetheorem}

\section{Conclusion, Related and Future Work}
\label{sec:conclusion}
In the last years, an impressive amount of papers has studied behavioral distances for both probabilistic and nondeterministic systems~(see, e.g.,~\cite{GJS90,DGJP04,vBW05,bblm:total-variation-markov,afs:linear-branching-metrics-quant,afs:linear-branching-metrics,flt:quantitative-spectrum}). The necessity of a general understanding of such metrics is not a mere intellectual whim but it is perceived also by researchers exploiting distances for differential privacy and quantitative information flow (see for instance~\cite{DBLP:conf/concur/ChatzikokolakisGPX14}). As far as we know, the first use of coalgebras for this purpose dates back to~\cite{vBW05}, where the authors consider systems and distance for a fixed endofunctor on $\PMet$. 
In~\cite{BBKK14}, we introduced the Kantorovich and Wasserstein approaches as a general way to define ``canonical liftings'' to $\PMet$ and behavioral distances by finality. These are usually branching-time, while many properties of interest for applications (see again~\cite{DBLP:conf/concur/ChatzikokolakisGPX14}) are usually expressed by means of distances on set of traces. In this paper, we have shown that the work developed in~\cite{BBKK14} can be fruitfully combined with~\cite{JSS15} to obtain various trace distances.

Among the several trace distances introduced in literature, it is worth to mention~\cite{bblm:total-variation-markov,afs:linear-branching-metrics-quant,afs:linear-branching-metrics,flt:quantitative-spectrum}. Similar to the trace distance we obtain
in \autoref{exa:final-coalgebra} for probabilistic automata is the one introduced in~\cite{bblm:total-variation-markov} for Semi-Markov chains
with residence time. In
\cite{afs:linear-branching-metrics-quant,afs:linear-branching-metrics}, both branching-time and linear-time distances are introduced for
\emph{metric transition systems}, namely Kripke structures where states are
associated with elements of a fixed (pseudo-)metric space $M$, that
would correspond to coalgebras of the form $X\to M\times
\mathcal{P}(X)$. In~\cite{BBKK14}, we have shown an example capturing branching-time distance for metric transition systems, 
but for linear distances we require a distributive law of the
form $\mathcal{P}(M\times \_)\Rightarrow M\times \mathcal{P}(\_)$, for
which we would need at least $M$ carrying an algebra for the monad
$\Powerset$. We also plan to investigate trace metrics in a Kleisli setting
\cite{HJS07}, where it might be easier to incorporate such examples.

There are two other direct consequences of our work that we did not explain in the main text, but that are important properties of the distances that we obtain (and, indeed, are mentioned in~\cite{DBLP:conf/concur/ChatzikokolakisGPX14} amongst the desiderata for ``good'' metrics).
First, the behavioral branching-distance for $\LiftedFunctor{F}\,\LiftedFunctor{T}$ provides an upper bound to the linear-distance $\LiftedFunctor{F}$, analogously to the well-known fact that bisimilarity implies trace equivalence. To see this, it is enough to observe that there is a functor from the category of $\LiftedFunctor{F}\,\LiftedFunctor{T}$-coalgebras to the one of $\LiftedFunctor{F}$-coalgebras mapping $c\colon X \to \LiftedFunctor{F}\,\LiftedFunctor{T}X$ into $c^{\sharp}\colon \LiftedFunctor{T}X \to \LiftedFunctor{F}\,\LiftedFunctor{T}X$.

Second, since the final map $\mean{-}$ is a morphism in $\EM(\LiftedFunctor{T})$, the behavioral distance for $\LiftedFunctor{F}$ is nonexpansive w.r.t. the operators of the monad $\LiftedFunctor{T}$. Nonexpansiveness with respect to some operators is a desirable property which has been studied, for instance in~\cite{DGJP04}, as a generalization of the notion of being a congruence for behavioral equivalence. Several researchers are now studying syntactic rule formats ensuring this and other sorts of compositionality (see e.g.~\cite{DBLP:journals/corr/GeblerT13} and the references therein) and we believe that our \autoref{cor:lifting-distr-law} may provide some helpful insights.

In this perspective, however, our results are still unsatisfactory if compared to what happens in the case of behavioral equivalences. From a fibrational point of view, one has a \emph{canonical lifting} to $\cat{Rel}$ (the category of relations and relation preserving morphisms) such that compositionality holds on the nose and distributive laws always lift \cite[Exercise 4.4.6]{Jacobs:coalg}. The forgetful functor $U\colon \PMet \to \Set$ is also a fibration~\cite{BBKK14}, but Kantorovich and Wasserstein liftings are not always so well-behaved. Fibrations might be useful also to guarantee soundness of up-to techniques~\cite{DBLP:journals/corr/BonchiPPR14} for behavioral distances that, hopefully, will lead to more efficient proofs and algorithms. 

Another interesting future work would be to show that Kantorovich and
Wasserstein liftings arise from some universal properties, i.e., that
they are the smallest and largest metric in some continuum of metrics
with certain properties. Here we would like to draw inspiration from
\cite{vB05} which characterizes the Giry monad via a universal
property on monad morphisms.

Finally, we would like to have an abstract understanding of the
Kan\-to\-ro\-vich-Ru\-bin\-stein duality.  Preliminary attempts suggest 
that this is very difficult: indeed the proof for the probabilistic
case relies on specific properties of distributions.
 
\bibliographystyle{alpha} 
\bibliography{calco15}
\appendix
\clearpage
\makeatletter{}
\setcounter{section}{15} 
\section{Proofs}
\label{sec:proofs}
Here we provide proofs for the soundness of our definitions (where needed), the stated theorems, propositions, lemmas, examples and also for all claims made in the in-between texts. If a theorem environment starts with the symbol $\circlearrowleft$ it has been stated in the main text and is repeated here for convenience of the reader (using the numbering from the main text). Otherwise it is a new statement which clarifies/justifies claims made in the main text and its number starts with \emph{P}.

\stepcounter{subsection}
\renewcommand{\sectionautorefname}{Appendix}
\renewcommand{\subsectionautorefname}{Appendix}

\subsection{\nameref{sec:preliminaries}}
For the upcoming proofs we will often use the following, alternative characterization of W3. 
~\medskip\\\begin{minipage}{.68\textwidth}
\begin{lemmaApx}[Weak Pullback Characterization of W3]
\label{lem:W3-weak-pb}
Let $F$ be an endofunctor on $\Set$ with evaluation function $\ev_F$ and $i \colon \set{0} \hookrightarrow [0,\top]$ be the inclusion function. For any set $X$ we denote the unique arrow into $\set{0}$ by $!_X\colon X \to \set{0}$. Then $\ev_F$ satisfies $\ev_F^{-1}[\set{0}] = Fi[F\set{0}]$ if and only if the diagram on the right is a weak pullback.
\end{lemmaApx}
\end{minipage}\hfill%
\begin{minipage}{.28\textwidth}
		\begin{diagram}
			\matrix(m)[matrix of math nodes, column sep=30pt, row sep=30pt,ampersand replacement=\&]{
				F\set{0} \& \set{0}\\
				F[0,\top] \& {[0,\top]}\\
			};
			\draw[->] (m-1-1) edge node[above]{$!_{F\set{0}}$} (m-1-2);
			\draw[->] (m-1-1) edge node[left]{$Fi$} (m-2-1);
			\draw[right hook->] (m-1-2) edge node[right]{$i$} (m-2-2);
			\draw[->] (m-2-1) edge node[above]{$\ev_F$} (m-2-2);
		\end{diagram}\end{minipage}

\begin{proof}
Commutativity of the diagram is equivalent to $\ev_F^{-1}[\set{0}] \supseteq  Fi[F\set{0}]$. Given a set $X$ and a function $f \colon X \to F[0,\top]$ as depicted below, we conclude again by commutativity ($i \circ !_X = \ev_F \circ f$) that $f(x) \in \ev_F^{-1}[\set{0}]$ for all $x \in X$.
\begin{center}
		\begin{diagram}
			\matrix(m)[matrix of math nodes, column sep=30pt, row sep=15pt]{
				X\\
				&F\set{0} & \set{0}\\
				&F[0,\top] & {[0,\top]}\\
			};
			\draw[->] (m-1-1) edge[bend left] node[above] {$!_X$} (m-2-3);
			\draw[->] (m-1-1) edge[bend right] node[left] {$f$} (m-3-2);
			\draw[->] (m-1-1) edge[dashed] node[above] {$\phi$} (m-2-2);
			\draw[->] (m-2-2) edge node[above]{$!_{F\set{0}}$} (m-2-3);
			\draw[->] (m-2-2) edge node[left]{$Fi$} (m-3-2);
			\draw[right hook->] (m-2-3) edge node[right]{$i$} (m-3-3);
			\draw[->] (m-3-2) edge node[above]{$\ev_F$} (m-3-3);
		\end{diagram}\end{center}
Now if $\ev_F^{-1}[\set{0}] \subseteq  Fi[F\set{0}]$ then for $f(x) \in \ev_F^{-1}[\set{0}]$ we can choose a (not necessarily unique) $x_0 \in F\set{0}$ such that $f(x) = Fi(x_0)$. If we define $\phi\colon X \to F\set{0}$ by $\phi(x)  = x_0$ then clearly $\phi$ makes the above diagram commute and thus we have a weak pullback.

Conversely if the diagram is a weak pullback consider the set $X = \ev_F^{-1}[\set{0}]$ and the function $f \colon \ev_F^{-1}[\set{0}] \hookrightarrow F[0,\top], f(x) = x$. Now for any $x \in \ev_F^{-1}[\set{0}]$ we have $Fi (\phi(x)) = (Fi \circ \phi)(x) = f(x) = x$, hence since $\phi(x) \in F\set{0}$ we have $x \in Fi[F\set{0}]$.
\end{proof}

\restatewithstyle{exa:distributions}{restatedefinition}
\begin{proof}
Weak pullback preservation, well-behavedness and the duality was already presented in \cite{BBKK14}. Here we just quickly check that indeed the infimum is a minimum: Let $\supp(P_1) \cup \supp(P_2) = \set{s_1,\dots,s_n}$ be the union of the finite supports of $P_1$ and $P_2$. Then define the following finitely many real numbers $p_{1i} := P_1(s_i)$, $p_{2j} := P_2(s_j)$, $d_{ij}:= d(s_i, s_j)$. Then the distance of $P_1$ and $P_2$ can be equivalently expressed as the following LP:
\begin{alignat*}{2}
\text{minimize }\quad  & \sum_{1 \leq i,j \leq n} d_{ij} \cdot x_{ij}   \\
\text{subject to }\quad %
& \sum_{1 \leq j \leq n} x_{ij} = p_{1i}&,\ & 1 \leq i \leq n\\
& \sum_{1 \leq i \leq n} x_{ij} = p_{2j}&,\ & 1 \leq j \leq n\\
& 0 \leq x_{ij} \leq 1, &\ & 1 \leq i,j \leq n
\end{alignat*}
whose feasible region is nonempty ($x_{ij}:=p_{1i} \cdot p_{2j}$ is in it) and bounded. Thus we indeed get an optimal solution $x_{ij}^*$ and can define the optimal coupling as $P^*(s_i, s_j) := x_{ij}^*$.
\end{proof}

\restatewithstyle{exa:product-bifunctor}{restatedefinition}
\begin{proof}
We adapt the proof given in \cite[Exa. 5.1]{BBKK14} to also include the discounted maximum (all other cases were covered there). First we show well-behavedness.
\begin{enumerate}
\item Let $f_i,g_i\colon X_i \to [0,\top]$ with $f_i \leq g_i$ be given. Then we have 
\begin{align*}
	\EvaluationFunctor{F}(f_1,f_2) = \max(c_1f_1,c_2f_2) \leq \max(c_1g_1,c_2g_2) = \EvaluationFunctor{F}(g_1,g_2)\,.
\end{align*}
\item Let $t:=(x_{11},x_{21},x_{12},x_{22})\in F([0,\top]^2,[0,\top]^2) = [0,\top]^2\times [0,\top]^2$. We have to show the inequality $d_e\big(\EvaluationFunctor{F}(\pi_1,\pi_1)(t), \EvaluationFunctor{F}(\pi_2, \pi_2)(t)) \leq \EvaluationFunctor{F}(d_e,d_e)(t)$. We observe that $\EvaluationFunctor{F}(d_e,d_e)(t) = \ev_F(d_e(x_{11},x_{21}), d_e(x_{12},x_{22}))$ and if we define $z_i = \ev_F(x_{i1},x_{i2}) = \max\set{c_1x_{i1}, c_2x_{i2}}$ then $d_e\big(\EvaluationFunctor{F}(\pi_1,\pi_1)(t), \EvaluationFunctor{F}(\pi_2, \pi_2)(t)) = d_e(z_1,z_2)$. We thus have to show the inequality 
\begin{equation}
	d_e\left(z_1,z_2\right) \leq \ev_F(d_e(x_{11},x_{21}), d_e(x_{12},x_{22}))\,.\label{eq:condtwo:product}
\end{equation}
If $z_1=z_2$ this is obviously true because $d_e(z_1, z_2) = 0$ and the rhs is non-negative. We now assume $z_1 > z_2$ (the other case is symmetrical). For $\infty = z_1 > z_2$ the inequality holds because then $x_{11} = \infty$ or $x_{12} = \infty$ and $x_{21},x_{22} < \infty$ (otherwise we would have $z_2 = \infty$) so both lhs and rhs are $\infty$. Thus we can now restrict to $\infty > z_1 > z_2$ where necessarily also $x_{11}, x_{12}, x_{21}, x_{22} < \infty$ (otherwise we would have $z_1 = \infty$ or $z_2=\infty$). According to \cite[Lemma P2.1]{BBKK14}, the inequality~\eqref{eq:condtwo:product} is equivalent to showing the two inequalities
\begin{align*}
	z_1 &\leq z_2 + \ev_F\big(d_e(x_{11},x_{21}), d_e(x_{12},x_{22})\big), \quad \text{and} \\ \ z_2 &\leq z_1 + \ev_F\big(d_e(x_{11},x_{21}), d_e(x_{12},x_{22})\big)\,.
\end{align*}
By our assumption ($\infty > z_1>z_2$) the second of these inequalities is satisfied, so we just have to show the first. If $z_1 = c_1x_{11}$ we have 
	\begin{align*}
		z_2 + \max\set{c_1d_e(x_{11}, x_{21}), c_2 d_e(x_{12},x_{22})} &\geq z_2 + c_1d_e(x_{11},x_{21}) = z_2 + c_1|x_{11}-x_{21}| \\
		&\geq z_2 + c_1(x_{11}-x_{21})= z_2 + c_1x_{11} - c_1x_{21} \\
		&= z_2 + z_1 - c_1x_{21} = z_1 + (z_2-c_1x_{21}) \geq z_1
		\end{align*}
	because $z_2 = \max\set{c_1x_{21},c_2x_{22}} > c_1x_{21}$ and therefore $(z_2-c_1x_{21}) \geq 0$. The same line of argument can be applied if $z_1 = c_2x_{12}$.
\item $F(i,i)[F(\set{0},\set{0})] = (i\times i)[\set{0} \times \set{0}] = \set{(0,0)}$ and $\ev_F^{-1}[\set{0}] = \set{(0,0)}$.
\end{enumerate}
We now prove that the Kantorovich-Rubinstein duality holds and simultaneously that the supremum (in the Kantorovich pseudometric) is a maximum and the infimum (of the Wasserstein pseudometric) is a minimum. Let $(X_1,d_1)$, $(X_1,d_2)$ be pseudometric spaces and let $t_i = (x_{i1},x_{i2}) \in F(X_1,X_2) = X_1 \times X_2$ be given. Their unique coupling is $t := ((x_{11},x_{21}),(x_{12},x_{2})) \in \Gamma_F(t_1,t_2)$ and we have $\EvaluationFunctor{F}(d_1,d_2)(t) = \max\{c_1d_1(x_1,y_1),c_2d_2(x_2,y_2)\}$. We define $f_i:=d_i(x_{1i},\_)$, which are nonexpansive due to \cite[Lemma 2.3]{BBKK14}. Then we clearly have $f_i(x_{1i}) = 0$ and moreover
\begin{align*}
	d_e\big(\EvaluationFunctor{F}(f_1,f_2)(t_1),\EvaluationFunctor{F}(f_1,f_2)(t_2)\big) &= d_e\left(\ev_F \big(f_1(x_{11}),f_2(x_{12})\big),\ev_F \big(f_1(x_{21}),f_2(x_{22})\big)\right)\\
	& = d_e\big(0, \max\set{c_1d_1(x_{11},x_{21}), c_2d_2(x_{12}, x_{22})}\big) \\
	&= \max\set{c_1d_1(x_{11},x_{21}), c_2d_2(x_{12}, x_{22})} = \EvaluationFunctor{F}(d_1,d_2)(t)
\end{align*}
Due to \cite[Proposition P.5.7]{BBKK14} we can now conclude that duality holds and both supremum and infimum are attained and equal to the above maximum.
\end{proof}

\subsection{\nameref{sec:compositionality}}
\subsubsection{Compositionality for Endofunctors}
We first collect a few simple observations that we will use in the upcoming proofs.
\begin{lemmaApx}
\label{lem:comp1}
Let $F,G$ be endofunctors on $\Set$ with evaluation functions $\ev_F, \ev_G$ and $a:= \langle G\pi_1, G\pi_2\rangle$ (i.e. the unique mediating arrow into the product) and $(X,d)$ an arbitrary pseudometric space. Then the following holds.
\begin{enumerate}
	\item $\EvaluationFunctor{G}d \ge \Wasserstein{G}{d} \circ a \ge \Kantorovich{G}{d} \circ a$
	\item $\forall t_1,t_2 \in FGX: \quad t \in \Couplings{FG}(t_1,t_2) \implies Fa(t) \in \Couplings{F}(t_1,t_2)$.
	\item If $F$ and $G$ preserve weak pullbacks then so does $FG$.
	\item For any $f \in \Set/[0,\top]$ we have $\EvaluationFunctor{FG}f = \EvaluationFunctor{F}(\EvaluationFunctor{G}f)$.
\end{enumerate}
\end{lemmaApx}
\begin{proof}
We first of all observe that $a$ is the unique mediating arrow into the product as indicated in the following diagram.
\begin{center}\begin{diagram}
	\matrix(m)[matrix of math nodes, column sep=80pt, row sep=40pt]{
		& G(X \times X)\\
		GX & GX \times GX & GX\\
	};

	\draw[->] (m-1-2) edge[bend right] node[left,xshift=-5pt]{$G\pi_1$} (m-2-1);
	\draw[->] (m-1-2) edge[bend left] node[right,xshift=5pt]{$G\pi_2$} (m-2-3);
	
	\draw[->] (m-2-2) edge node[below]{$\pi_1$} (m-2-1);
	\draw[->] (m-2-2) edge node[below]{$\pi_2$} (m-2-3);
	
	\draw[->] (m-1-2) edge node[right]{$a = \tuple{G\pi_1,G\pi_2}$} (m-2-2);
\end{diagram}\end{center}
\begin{enumerate}
\item Let $s \in G(X\times X)$ and define $s_i:=G\pi_i^X(s) = \pi_i^{GX} \circ a(s)$. Then by definition $s \in \Couplings{G}(s_1,s_2)$ and we conclude $\EvaluationFunctor{G}d(s) \geq \inf\{\EvaluationFunctor Gd(s') \mid s' \in \Couplings{G}(s_1,s_2)\} = \Wasserstein{G}{d}(s_1,s_2) = \Wasserstein{G}{d}(\pi_1^{GX} \circ a(s),\pi_2^{GX} \circ a(s)) = \Wasserstein{G}{d} \circ a (s)$. Since we always have $\Wasserstein{G}{d} \geq \Kantorovich{G}{d}$ as shown in \cite{BBKK14}, the statement follows.
\item $F\pi_i^{GX} (Fa(t)) = F(\pi_i^{GX} \circ a)(t) = F(G\pi_i^X)(t) = FG\pi_i^X = t_i$. 
\item This is indeed clear by definition.
	\item Let $f\colon X \to [0,\top]$, then $\EvaluationFunctor{FG}f = \ev_{FG} \circ FGf = \ev_F \circ F\ev_G \circ FGf = \ev_F \circ F(\ev_G \circ Gf) = \EvaluationFunctor{F}(\EvaluationFunctor{G}f)$.\qedhere\end{enumerate}
\end{proof}

\begin{lemmaApx}
\label{lem:comp:evfct}
Let $F$,$G$ be functors with evaluation functions $\ev_F$ and $\ev_G$ and define $\ev_{FG} := \ev_F \circ F\ev_G$. Then the following holds.
\begin{enumerate}
\item If $\EvaluationFunctor{F}$ and $\EvaluationFunctor{G}$ are monotone (Condition W1), then so is $\EvaluationFunctor{FG}$.
	\item If $G$ preserves weak pullbacks, $\ev_G$ is well-behaved and $\EvaluationFunctor{F}$ is monotone then $\ev_{FG}$ satisfies Condition W2 of well-behavedness.
	\item If $F$ preserves weak pullbacks and $\ev_F,\ev_G$ satisfy Condition W3 of well-behavedness, then also $\ev_{FG}$ satisfies Condition W3 of well-behavedness.	
\end{enumerate}
\end{lemmaApx}
\begin{proof}
\begin{enumerate}
\item Let $f,g\colon X \to [0,\top]$ with $f \leq g$, then by monotonicity of $\ev_G$ we have $\EvaluationFunctor{G}f \leq \EvaluationFunctor{G}g$ and using monotonicity of $\ev_F$ we get $\EvaluationFunctor{FG}f = \EvaluationFunctor{F}(\EvaluationFunctor{G}f) \leq \EvaluationFunctor{F}(\EvaluationFunctor{G}g) = \EvaluationFunctor{FG}g$. 
	
	\item Let $t \in FG([0,\top]^2)$ and define $t_i := FG\pi_i(t) \in FG[0,\top]$. By definition $t \in \Couplings{FG}(t_1,t_2)$ so \autoref{lem:comp1} tells us $Fa(t) \in \Couplings{F}(t_1,t_2)$ for $a:=\langle G\pi_1, G\pi_2\rangle$. Moreover, since $\ev_G\colon (G[0,\top],\Kantorovich{G}{d_e}) \nonexpansiveTo ([0,\top],d_e)$ is nonexpansive (by definition of the Kantorovich pseudometric), we can apply \cite[Prop. P.4.2]{BBKK14} to obtain the inequality 
	\begin{align*}
		d_e(\ev_{FG}(t_1), \ev_{FG}(t_2)) = d_e(\EvaluationFunctor{F}\ev_G(t_1), \EvaluationFunctor{F}\ev_G(t_2)) \leq \EvaluationFunctor{F}\Kantorovich{G}{d_e}(Fa(t)) = \EvaluationFunctor{F}(\Kantorovich{G}{d_e} \circ a)(t)\,.
		\end{align*}
		By  \autoref{lem:comp1} we have $\Kantorovich{G}{d_e} \circ a \leq \EvaluationFunctor{G}{d_e}$ and using monotonicity of $\EvaluationFunctor{F}$ we can continue our inequality with $\EvaluationFunctor{F}(\Kantorovich{G}{d_e} \circ a)(t) \leq \EvaluationFunctor{F}(\EvaluationFunctor{G}{d_e})(t) = \EvaluationFunctor{FG}d_e(t)$ which concludes the proof!
		\item Using \autoref{lem:W3-weak-pb} we just have to show that the following diagram is a weak pullback.
		\begin{center}
		\begin{diagram}
			\matrix(m)[matrix of math nodes, column sep=80pt, row sep=20pt]{
				FG\set{0} & F\set{0} & \set{0}\\
				FG[0,\top] & F{[0,\top]} & {[0,\top]}\\
			};
			\draw[->] (m-1-1) edge node[below]{$F!_{G\set{0}}$} (m-1-2);
			\draw[->] (m-1-1) edge node[left]{$FGi$} (m-2-1);
			\draw[->] (m-1-2) edge node[right]{$Fi$} (m-2-2);
			\draw[->] (m-2-1) edge node[above]{$F\ev_G$} (m-2-2);
			
			\draw[->] (m-1-2) edge node[below]{$!_{F\set{0}}$} (m-1-3);
			\draw[right hook->] (m-1-3) edge node[right]{$i$} (m-2-3);
			\draw[->] (m-2-2) edge node[above]{$\ev_F$} (m-2-3);
			
			\draw[->] (m-1-1) edge[bend left=20] node[above]{$!_{FG\set{0}}$} (m-1-3);
			\draw[->] (m-2-1) edge[bend right=20] node[below]{$\ev_{FG}$} (m-2-3);
		\end{diagram}
		\end{center}
		\autoref{lem:W3-weak-pb} tells us that the right square is a weak pullback and since $F$ preserves weak pullbacks also the left square is. The outer part is necessarily a weak pullback again yielding by \autoref{lem:W3-weak-pb} that $\ev_{FG}$ satisfies the third condition.\qedhere
\end{enumerate}
\end{proof}

\restate{prop:comp:well-behaved}
\begin{proof}
This is an immediate corollary of \autoref{lem:comp:evfct}.
\end{proof}

\noindent To prove our compositionality criteria, we use the following results.
\begin{lemmaApx}
\label{lem:compositionality2}
Let $F$, $G$ be endofunctors on $\Set$ with evaluation functions $\ev_F\colon F[0,\top]\to [0,\top]$, $\ev_G\colon G[0,\top] \to [0,\top]$. We define $\ev_{FG} := \ev_F \circ F\ev_G$. Then the following holds for every pseudometric space $(X,d)$.
\begin{enumerate}
	\item $\Kantorovich{FG}{d} \leq \Kantorovich{F}{(\Kantorovich{G}{d})}$.\label{item:leqKantorovich}
	\item If $F$ and $G$ preserve weak pullbacks and $\ev_F, \ev_G$ are well-behaved then $\Wasserstein{FG}{d} \geq \Wasserstein{F}{(\Wasserstein{G}{d})}$.\label{item:geqWasserstein}
	\item If for all $t_1,t_2 \in FGX$ there is a function $\nabla (t_1,t_2)\colon \Couplings{F}(t_1, t_2) \to \Couplings{FG}(t_1,t_2)$ such that $\EvaluationFunctor{FG}d \circ \nabla(t_1,t_2) = \EvaluationFunctor{F}\Wasserstein{G}{d}$ then $\Wasserstein{FG}{d} \leq \Wasserstein{F}{(\Wasserstein{G}{d})}$.\label{prop:comp:nabla}
\end{enumerate}
\end{lemmaApx}
\begin{proof}
Let $t_1,t_2 \in FGX$. 
\begin{enumerate}
	\item Recall that $\Kantorovich{G}{d}$ is the smallest pseudometric such that for every nonexpansive function $f\colon (X,d) \nonexpansiveTo ([0,\top],d_e)$ also $\EvaluationFunctor{G}f \colon (GX,\Kantorovich{G}{d}) \nonexpansiveTo ([0,\top],d_e)$ is nonexpansive (see remark after \cite[Def. 3.1]{BBKK14}). Moreover, $\EvaluationFunctor{FG}f = \EvaluationFunctor{F}(\EvaluationFunctor{G}f)$ by \autoref{lem:comp1}. Thus 
	\begin{align*}
		\Kantorovich{FG}{d}(t_1,t_2) &= \sup\set{d_e\big(\EvaluationFunctor{FG}f(t_1),\EvaluationFunctor{FG}f(t_2)\big) \mid f \colon (X,d) \nonexpansiveTo ([0,\top],d_e)}\\
		& = \sup\set{d_e\big(\,\EvaluationFunctor{F}(\EvaluationFunctor{G}f)(t_1),\EvaluationFunctor{F}(\EvaluationFunctor{G}f)(t_2)\big) \mid f \colon (X,d) \nonexpansiveTo ([0,\top],d_e)}\\
		& \leq \sup\set{d_e\big(\,\EvaluationFunctor{F}(g)(t_1),\EvaluationFunctor{F}(g)(t_2)\big) \mid g \colon (GX,\Kantorovich{G}{d}) \nonexpansiveTo ([0,\top],d_e)} \\
		&= \Kantorovich{F}{(\Kantorovich{G}{d})}(t_1,t_2)
	\end{align*}
	
	\item \autoref{lem:comp1} tells us $\EvaluationFunctor{G}d \geq \Wasserstein{G}{d} \circ a$ and for any coupling $t \in \Couplings{FG}(t_1,t_2)$ we have $Fa(t) \in \Couplings{F}(t_1,t_2)$. Using these facts and the monotonicity of $\EvaluationFunctor{F}$ we obtain:
	\begin{align*}
		\Wasserstein{FG}{d}(t_1,t_2) &= \inf\set{\EvaluationFunctor{FG}d(t) \mid t\in \Couplings{FG}(t_1,t_2)} = \inf\set{\EvaluationFunctor{F}(\EvaluationFunctor{G}d)(t) \mid t\in \Couplings{FG}(t_1,t_2)}\\
		& \geq  \inf\set{\EvaluationFunctor{F}(\Wasserstein{G}{d} \circ a)(t) \mid t\in \Couplings{FG}(t_1,t_2)}\\
		& =  \inf\set{\EvaluationFunctor{F}\Wasserstein{G}{d} \big(Fa(t)\big) \mid t\in \Couplings{FG}(t_1,t_2)} \\
		& \geq  \inf\set{\EvaluationFunctor{F}\Wasserstein{G}{d} (t') \mid t'\in \Couplings{F}(t_1,t_2)} = \Wasserstein{F}{(\Wasserstein{G}{d})}(t_1,t_2)
	\end{align*}

\item Using $\nabla(t_1,t_2)$ we compute
\begin{align*}
\Wasserstein{FG}{d}(t_1,t_2) &=  \inf{\set{\EvaluationFunctor{FG}d (t') \mid t'\in \Couplings{FG}(t_1,t_2)} }\\
& \leq \inf{\set{\EvaluationFunctor{FG}d \big(\nabla(t_1,t_2)(t)\big) \mid t\in \Couplings{F}(t_1,t_2)} } \\
&= \inf{\set{\EvaluationFunctor{F}\Wasserstein{G}{d} (t) \mid t\in \Couplings{F}(t_1,t_2)} } =\Wasserstein{F}{(\Wasserstein{G}{d})}(t_1,t_2) \,.\qedhere
\end{align*}
\end{enumerate}
\end{proof}

\noindent With this result at hand we can prove
\restate{prop:comp}
\begin{proof}
From \autoref{lem:compositionality2}.\ref{item:geqWasserstein} we know $\Wasserstein{FG}{d} \geq \Wasserstein{F}{(\Wasserstein{G}{d})}$. By our requirement we have a function $\gamma\colon GX \times GX \to G(X\times X)$, such that $\Wasserstein{G}{d} = \EvaluationFunctor{G}d \circ \gamma$. Given $t_1,t_2 \in FGX$ and $t \in \Couplings{F}(t_1,t_2)$, we define $\nabla(t_1,t_2)(t) = F\gamma(t)$, then this satisfies the conditions of  \autoref{lem:compositionality2}.\ref{prop:comp:nabla}. First, we have $F\gamma(t) \in \Couplings{FG}(t_1,t_2)$ because $FG\pi_i^X(F\gamma(t)) = F(G\pi_i^X \circ \gamma)(t) = F\pi_i^{GX}(t) = t_i$. Moreover
\begin{align*}
\EvaluationFunctor{FG}d\big(F\gamma(t)\big) &= \ev_{FG} \circ F\big(Gd \circ \gamma(t)\big) = \ev_{F} \circ F\ev_G \circ F(Gd \circ \gamma)(t) \\
&=\ev_{F} \circ F\big(\EvaluationFunctor{G}d \circ \gamma\big)(t)  = \ev_{F} \circ F\Wasserstein{G}{d} (t) = \EvaluationFunctor{F}\Wasserstein{G}{d}(t)\,.\qedhere
\end{align*}
\end{proof}

\restatewithstyle{exa:comp}{restatedefinition}
\begin{proof}
We just have to prove the second claim. We already know from \autoref{lem:compositionality2}.\ref{item:geqWasserstein} that 
\begin{align}
	\Wasserstein{\PowersetFinite\PowersetFinite}{d} \geq \Wasserstein{\PowersetFinite}{\left(\Wasserstein{\PowersetFinite}{d}\right)}\label{eq:powfin0}
\end{align} holds. We now show that we always have equality. Let $(X,d)$ be a pseudometric space and $T_1,T_2 \in \PowersetFinite\PowersetFinite X$. We distinguish three cases:\\
\emph{Case 1}: If $T_1=T_2=\emptyset$ we know by reflexivity that both values are $0$.\\
\emph{Case 2}: If $T_1=\emptyset\not=T_2$ or $T_1\not=\emptyset=T_2$ we know from \cite{BBKK14} that $\Couplings{\PowersetFinite}(T_1,T_2) = \emptyset$ and therefore $\Wasserstein{\PowersetFinite}{\left(\Wasserstein{\PowersetFinite}{d}\right)}(T_1,T_2) = \top$ and thus \eqref{eq:powfin0} is an equality. \\
\emph{Case 3}: Let $T_1,T_2 \not=\emptyset$. We know from \cite{BBKK14} that we have an optimal coupling $T^* \in \Couplings{\PowersetFinite}(T_1,T_2)$, say $T^* = \set{(V_{j1}, V_{j2}) \in \PowersetFinite X \times \PowersetFinite X \mid j \in J}$ for a suitable index set $J$. Then $T_i = \PowersetFinite \pi_i (T^*) = \pi_i[T^*] = \set{\pi_i((V_{j1}, V_{j2})) \mid j \in J} = \set{V_{ji} \mid j \in J}$. By optimality:
\begin{align}
	\Wasserstein{\PowersetFinite}{\left(\Wasserstein{\PowersetFinite}{d}\right)}(T_1,T_2) = \EvaluationFunctor{\PowersetFinite}{\Wasserstein{\PowersetFinite}{d}}(T^*) = \max \Wasserstein{\PowersetFinite}{d}[T^*] = \max_{j \in J} \Wasserstein{\PowersetFinite}{d}(V_{j1}, V_{j2})\,. \label{eq:powfin1}
\end{align}
Again we will make a case distinction:
\begin{itemize}
\item If there is a $j' \in J$ such that $\Couplings{\PowersetFinite}(V_{j'1}, V_{j'2}) = \emptyset$, we have $\Wasserstein{\PowersetFinite}{d}(V_{j'1}, V_{j'2}) = \top $ and using \eqref{eq:powfin1} also $\Wasserstein{\PowersetFinite}{\left(\Wasserstein{\PowersetFinite}{d}\right)}(T_1,T_2) = \top$ which again shows that \eqref{eq:powfin0} is an equality. 
\item Otherwise we can take optimal couplings $V_j^* \in \Couplings{\PowersetFinite}(V_{j1}, V_{j2})$. Continuing \eqref{eq:powfin1} we have
\begin{align}
	\Wasserstein{\PowersetFinite}{\left(\Wasserstein{\PowersetFinite}{d}\right)}(T_1,T_2) = \max_{j \in J} \EvaluationFunctor{\PowersetFinite}d(V_j^*) = \max_{j \in J} \max d[V_j^*] \label{eq:powfin2}
\end{align}

Then we define $T := \set{V_j^* \mid j \in J} \subseteq \PowersetFinite\PowersetFinite(X \times X)$. We calculate for $\pi_i \colon X \times X \to X$
\begin{align*}
\PowersetFinite\PowersetFinite \pi_i (T) = \PowersetFinite\pi_i[T] = \set{\PowersetFinite\pi_i(V_j^*) \mid j \in J} = \set{V_{ji} \mid j \in J} = T_i
\end{align*}
and thus $T \in \Couplings{\PowersetFinite\PowersetFinite}(T_1,T_2)$. Moreover we have
\begin{align}
\Wasserstein{\PowersetFinite\PowersetFinite}{d}(T_1,T_2)  &\leq \EvaluationFunctor{\PowersetFinite\PowersetFinite}d(T) =\max\left(\PowersetFinite \max\left(\PowersetFinite\PowersetFinite(T)\right)\right) \nonumber\\
&= \max \left( \max \left[ \PowersetFinite d [T] \right] \right) = \max \left( \max\set{ d[V_j^*]\mid j \in J} \right) \nonumber\\
&= \max_{j \in J}\max d[V_j^*] \label{eq:powfin3}
 \end{align}
thus using this, \eqref{eq:powfin2} and \eqref{eq:powfin0} we conclude that 
\begin{align*}
\Wasserstein{\PowersetFinite\PowersetFinite}{d}(T_1,T_2) \leq \max_{j \in J}\max d[V_j^*] = \Wasserstein{\PowersetFinite}{\left(\Wasserstein{\PowersetFinite}{d}\right)}(T_1,T_2) \leq \Wasserstein{\PowersetFinite\PowersetFinite}{d}(T_1,T_2)
\end{align*}
which proves equality. \qedhere
\end{itemize}
\end{proof}

\noindent To verify the claims made in \autoref{exa:inputfunctor} we need the following intermediary result.
\begin{lemmaApx}
\label{lem:max-sum}
For finite $A$ and functions $f,g\colon A \to [0,\infty]$ we have
\begin{enumerate}
\item $d_e\big(\max_{a \in A}f(a),\max_{a \in A}g(a)\big) \leq \max_{a \in A} d_e\big(f(a), g(a)\big)$.
\item $d_e\left(\sum_{a \in A}f(a),\sum_{a \in A}g(a)\right) \leq \sum_{a \in A} d_e\big(f(a), g(a)\big)$.
\end{enumerate}
\end{lemmaApx}
\begin{proof}
\begin{enumerate}
\item Let $a_f \in \argmax_{a \in A} f(a)$ and $a_g \in \argmax_{a \in A} g(a)$, i.e. $a_f = \max_{a \in A}f$ and $a_g = \max_{a \in A}g$. If $f(a_f) = g(a_g)$ the lhs is $0$ and the inequality is satisfied. From here we assume wlog $f(a_f) > g(a_g)$. Now if $f(a_f) = \infty$, the lhs is $\infty$ but also $\max_{a \in A} d_e(f(a), g(a)) \geq d_e(f(a_f), g(a_f)) = \infty$. Finally, for $f(a_f) < \infty$ we have $g(a_f) \leq g(a_g)$ and thus $d_e(f(a_f), g(a_g)) = f(a_f) - g(a_g) \leq f(a_f) - g(a_f) \leq \max_{a \in A}d_e(f(a), g(a))$.
\item Let $s_f := \sum_{a \in A} f(a)$ and $s_g := \sum_{g \in A} f(a)$. If $s_f = s_g$ the lhs is $0$ and the inequality is satisfied. From here we assume wlog $s_f > s_g$. Now if $s_f=\infty$, the lhs is $\infty$ but we also must have an $a' \in A$ such that $f(a') = \infty$ (otherwise $s_f < \infty$) and thus $\sum_{a \in A} d_e(f(a), g(a)) \geq d_e(f(a'), g(a')) = \infty$. Finally, for $s_f< \infty$ we have $d_e(s_f, s_g) = s_f - s_g =  \sum_{a \in A} f(a) - \sum_{a \in A} g(a) = \sum_{a \in A} \left(f(a)-g(a)\right) \leq \sum_{a \in A} |f(a)-g(a)| = \sum_{a \in A} d_e(f(a), g(a))$.\qedhere
\end{enumerate}
\end{proof}

\restatewithstyle{exa:inputfunctor}{restatedefinition}
\begin{proof}
We first show that the functor $F := \_^A$ on $\Set$ preserves pullbacks. If we have a pullback in $\Set$ as indicated in the left of the diagram below, then we have to show that the right diagram is a pullback. 
\begin{center}\begin{diagram}
	\matrix(m)[matrix of math nodes, column sep=1.5cm, row sep=10pt]{
		P_{} & X_1 && P^A_{} & X_1^A\\
		X_2 & Y_{} && X_2^A & Y^A_{}\\
	};
	\draw[->] (m-1-1) edge node[above]{$p_1$} (m-1-2);
	\draw[->] (m-1-1) edge node[left]{$p_2$} (m-2-1);
	\draw[->] (m-1-2) edge node[right]{$f_1$} (m-2-2);
	\draw[->] (m-2-1) edge node[below]{$f_2$} (m-2-2);
	
	\draw[->] (m-1-4) edge node[above]{$p_1^A$} (m-1-5);
	\draw[->] (m-1-4) edge node[left]{$p_2^A$} (m-2-4);
	\draw[->] (m-1-5) edge node[right]{$f_1^A$} (m-2-5);
	\draw[->] (m-2-4) edge node[below]{$f_2^A$} (m-2-5);
\end{diagram}\end{center}
We consider the canonical pullback: $P := \set{(x_1, x_2) \in X_1 \times X_2 \mid f_1(x_1) = f_2(x_2)}$ along with $p_i:= \pi_i|_{P}$ and
\begin{align*}
	P' &:= \set{(g_1,g_2) \in X_1^A \times X_2^A \mid f_1^A(g_1) = f_2^A(g_2)} \\
	&\cong \set{\langle g_1,g_2\rangle \in (X_1\times X_2)^A \mid \forall a \in A.f_1\big(g_1(a)\big) = f_2\big(g_2(a)\big)}\\
	&\cong \set{\langle g_1,g_2\rangle \in (X_1\times X_2)^A \mid \forall a \in A.\Big(\big(g_1(a), g_2(a)\big) \in P\Big)} \cong P^A
\end{align*}
which completes the proof of weak pullback preservation. We now show that the evaluation functions are well-behaved. For $f\colon X \to [0,\top]$ we have $\EvaluationFunctor{F}f = \ev_F \circ f^A$ i.e. applying it to $g\in X^A$ yields $\max{a \in A} f(g(a))$ or $\sum_{a \in A} f(g(a))$. 
	\begin{well-behaved-axioms}
	\item For $f_1,f_2  \colon X \to [0,\top]$ with $f_1 \leq f_2$ we obviously also have $\EvaluationFunctor{F}f_1 \leq \EvaluationFunctor{F}f_2$.
	\item Let $t \in ([0,\top]^2)^A$ and $t_i := \pi_i^A(t)$, i.e. necessarily $t = \langle t_1, t_2\rangle$. We have to show 
	\begin{align*}
		d_e\big(\ev_F(t_1), \ev_F(t_2)\big) \leq \EvaluationFunctor{F}d_e(t) = \ev_F \big(d_e^A(t)\big) = \ev_F (d_e \circ t) = \ev_F(d_e \circ \tuple{t_1,t_2})\,.
	\end{align*}
	which for our evaluation functions follows from \autoref{lem:max-sum} with $f=t_1$, $g=t_2$.

	\item We have $\ev_F^{-1}[\set{0}] = \set{g \colon A \to [0,\top] \mid \ev_F(g) =0}$. Clearly for both functions this is the case only if $r$ is the constant $0$-function. Since $\set{0}$ is a final object in $\Set$, there is a unique function $z \colon A \to \set{0}$. Thus $Fi[F\set{0}] = i^A[\set{0}^A] = \set{i^A(z)} = \set{i \circ z}$ and clearly $i \circ z \colon A \to [0,\top]$ is also the constant $0$-function. 
	\end{well-behaved-axioms}
Now if we have $s_1,s_2 \in X^A$ their unique coupling is $s:=\tuple{s_1,s_2} \colon A \to X \times X$. Moreover $\EvaluationFunctor{F}d(s) = \ev_F ( d^A(s)) = \ev_F(\lambda a. d(\tuple{s_1,s_2})) = \ev_F(\lambda a.d(s_1(a), s_2(a))$ and using the two different evaluation functions we obtain the given pseudometrics.
\end{proof}
\subsubsection{Compositionality for Multifunctors}
We conclude this section with a more detailed presentation on how our theory extends to multifunctors. 

For $n \in \N$ we denote by $[n] := \{1,\dots,n\} \subseteq \N$ the set of all positive natural numbers less than or equal to $n$. Now let $n_i \in \N$ for all $i \in[n]$ and $F\colon \Set^n \to \Set$ and $G_i\colon \Set^{n_i} \to \Set$ (for $i \in [n]$) be multifunctors with evaluation functions $\ev_F\colon F({[0,\top]}^n) \to [0,\top]$ and $\ev_{G_i}\colon G_i([0,\top]^{n_i}) \to [0, \top]$. We define $N:=\sum_{i=1}^n n_i$ and define the functor 
\begin{align*}
	H := F \circ \prod_{i=1}^nG_i = F \circ (G_1 \times G_2 \times \dots \times G_n) \colon \Set^N \to \Set
\end{align*}
Then we can define the evaluation function $\ev_H\colon H([0,\top]^N) \to [0,\top]$ by
\begin{align*}
	\ev_H := \ev_F \circ F(\ev_{G_1}, \ev_{G_2}, \dots, \ev_{G_n})\,.
\end{align*}
In this setting, compositionality for the Wasserstein lifting means that whenever we have $N$ pseudometric spaces $(X_i,d_i)$ the pseudometric $\Wasserstein{H}{(d_1,\dots, d_N)}$ is equal to
\begin{align*}
	\Wasserstein{F}{\left(\Wasserstein{G_1}{(d_1,\dots,d_{n_1})}, \Wasserstein{G_2}{(d_{n_1+1},\dots,d_{n_1+n_2})}\dots, \Wasserstein{G_n}{(d_{N-n_n+1}, \dots, d_N)}\right)}\,.
\end{align*}
In the examples in this paper we will just have the following two cases:
\begin{enumerate}
\item $n=1$, $n_1=2$ so that $F\colon \Set \to \Set$ is an endofunctor with evaluation function $\ev_F\colon F[0,\top] \to [0,\top]$ and $G \colon \Set^2 \to \Set$ is a bifunctor with evaluation function $\ev_G\colon G([0,\top],[0,\top]) \to [0,\top]$. Then we have $N=n_1=2$ and obtain the bifunctor $H = F \circ G \colon \Set^2 \to \Set$ with evaluation $\ev_H = \ev_F \circ F\ev_G \colon FG([0,1],[0,1]) \to [0,1]$. Compositionality means that for an two pseudometric spaces $(X_1,d_1)$, $(X_2,d_2)$ we have $\Wasserstein{H}{(d_1,d_2)} = \Wasserstein{F}{(\Wasserstein{G}{(d_1,d_2)})}$.
\item $n=2$, $n_1=n_2 = 1$ so that $F\colon \Set^2 \to \Set$ is a bifunctor with evaluation function $\ev_F \colon F([0,\top],[0,\top]) \to [0,\top]$ and $G_1,G_2\colon \Set \to \Set$ are endofunctors with evaluations $\ev_{G_i}\colon G_i[0,\top] \to [0,\top]$. Then we have $N=n_1+n_2 = 1 + 1=2$ and obtain the bifunctor $H = F \circ (G_1 \times G_2) \colon \Set^2 \to \Set$ with evaluation $\ev_H = \ev_F \circ F(\ev_{G_1}, \ev_{G_2})\colon F(G_1[0,\top],G_2[0,\top]) \to [0,\top]$. Compositionality means that for an two pseudometric spaces $(X_1,d_1)$, $(X_2,d_2)$ we have $\Wasserstein{H}{(d_1,d_2)} = \Wasserstein{F}{(\Wasserstein{G}{d_1},\Wasserstein{G_2}{d_2})}$.
\end{enumerate}

\noindent The results presented for endofunctors work analogously in the multifunctor case (the proofs can be transferred almost verbatim), so we do not explicitly present them here.

\restatewithstyle{exa:machinebifunctor}{restatedefinition}
\begin{proof}
We first compute the composed evaluation functions. Let $(o,s) \in [0,1] \times [0,1]^A$, then
\begin{align*}
	\ev_M(o,s) &= \ev_P \circ P(\id_{[0,\top]}, \ev_I)  (o,s) = \ev_P \circ (\id_{[0,\top]} \times \ev_I) (o,s) = \ev_P\Big(o,\ev_I(s)\Big)
\end{align*} 
For the first case we thus have $\ev_M(o,s) = \max\set{c_1o, c_2\max_{a\in A}s(a)}$ and for the second $\ev_M(o,s) = c_1o + c_2{|A|}^{-1} \sum_{a \in A} d(s_1(a),s_2(a))$ as claimed. Given $t_1,t_2 \in M(B,X)$ with $t_i = (b_i,s_i) \in B \times X^A$ we take their unique coupling $t := (b_1,b_2,\tuple{s_1,s_2})$ to compute for pseudometrics $d_B$ on $B$ and $d$ on $X$:
\begin{align*}
	\Wasserstein{M}{(d_B,d)}(t_1,t_2) &= \EvaluationFunctor{M}(d_B,d) (t) = \ev_M \circ M(d_B,d) (t) \\
	&= \ev_M \circ \left(d_B \times d^A\right) \left(b_1,b_2,\tuple{s_1,s_2}\right) \\
	&=\ev_M\Big(d_B(b_1,b_2), \lambda a.d\big(s_1(a), s_2(a)\big)\Big)\,.
\end{align*}
Now if we take the two evaluation functions from above, we obtain the Wasserstein pseudometrics which are given in the example.
\end{proof}

\restatewithstyle{exa:m2}{restatedefinition}
\begin{proof}
We first prove that the bifunctor and endofunctor liftings coincide. Given $t_1,t_2 \in 2 \times X^A$, say $t_i = (o_i,s_i)$, 
their unique $M$-coupling is $t = (o_1,o_2,\tuple{s_1,s_2})$. 

If $o_1 \not = o_2$ no $M_2$-coupling of $t_1,t_2$ exists so we have $\Wasserstein{M_2}{d}{(t_1,t_2)} = \top$ but also $d_2(o_1,o_2) = \top$ so $\Wasserstein{M}{(d_2,d)}(t_1,t_2) = \EvaluationFunctor{M}(d_2,d)(t)\geq c_1d_2(o_1,o_2) \geq \top$. 

If $o_1 = o_2$ the unique $M_2$-coupling of $t_1,t_2$ is $t'=(o_1, \tuple{s_1,s_2})$ and $d_2(o_1,o_2) = 0$ thus
\begin{align*}
	\Wasserstein{M}{(d_2,d_e)}(t_1,t_2) &= \EvaluationFunctor{M}(d_2,d)(t) = c_2\ev_I\Big(\lambda a.d\big(s_1(a),s_2(a)\big)\Big) \\
	&= \ev_{M_2} \Big(o_1, \lambda a.d\big(s_1(a),s_2(a)\big)\Big) = \ev_{M_2}\Big(\big(\id_2 \times d^A\big) \big(o_1,\tuple{s_1,s_2}\big) \Big) \\
	&= \EvaluationFunctor{M_2}(t') = \Wasserstein{M_2}{d}(t_1,t_2)\,.
\end{align*}

\noindent For compositionality we adapt the proof of \autoref{exa:comp}. We know from \autoref{lem:compositionality2}.\ref{item:geqWasserstein} that 
\begin{align}
	\Wasserstein{\PowersetFinite M_2}{d} \geq \Wasserstein{\PowersetFinite}{\left(\Wasserstein{M_2}{d}\right)}\label{eq:m2:0}
\end{align} holds. We now show that we always have equality. Let $(X,d)$ be a pseudometric space and $T_1,T_2 \in \PowersetFinite M_2 X = \PowersetFinite (2 \times X^A)$. We distinguish three cases:\\
\emph{Case 1}: If $T_1=T_2=\emptyset$ we know by reflexivity that both values are $0$.\\
\emph{Case 2}: If $T_1=\emptyset\not=T_2$ or $T_1 \not=\emptyset=T_2$ we know from \cite{BBKK14} that $\Couplings{\PowersetFinite}(T_1,T_2) = \emptyset$ and therefore $\Wasserstein{\PowersetFinite}{\left(\Wasserstein{M_2}{d}\right)}(T_1,T_2) = \top$ and thus \eqref{eq:m2:0} is an equality. \\
\emph{Case 3}: Let $T_1,T_2 \not=\emptyset$. We know from \cite{BBKK14} that we have an optimal coupling $T^* \in \Couplings{\PowersetFinite}(T_1,T_2)$, say $T^* = \set{\big((o_{j1},s_{j1}),(o_{j2},s_{j2})\big) \in M_2X \times M_2X \mid j \in J}$ for a suitable index set $J$. Then using $\pi_i \colon M_2X \times M_2X \to M_2X$ we have $T_i = \PowersetFinite \pi_i (T^*) = \pi_i[T^*] = \set{\pi_i\big((o_{j1},s_{j1}),(o_{j2},s_{j2})\big) \mid j \in J} = \set{(o_{ji}, s_{ji}) \mid j \in J}$. By optimality:
\begin{align}
	\Wasserstein{\PowersetFinite}{\left(\Wasserstein{M_2}{d}\right)}(T_1,T_2) &= \EvaluationFunctor{\PowersetFinite}{\Wasserstein{M_2}{d}}(T^*) = \max \Wasserstein{M_2}{d}[T^*] \nonumber\\
	&= \max_{j \in J} \Wasserstein{M_2}{d}\big((o_{j1},s_{j1}),(o_{j2},s_{j2})\big)\,. \label{eq:m2:1}
\end{align}
Again we will make a case distinction:
\begin{itemize}
\item If there is a $j' \in J$ such that $\Couplings{M_2}\big((o_{j'1},s_{j'1}),(o_{j'2},s_{j'2})\big) = \emptyset$ (iff $o_{j'1} \not = o_{j'2}$), we have $\Wasserstein{M_2}{d}\big((o_{j1},s_{j1}),(o_{j2},s_{j2})\big) = \top $ and using \eqref{eq:m2:1} also $\Wasserstein{\PowersetFinite}{\left(\Wasserstein{M_2}{d}\right)}(T_1,T_2) = \top$ which again shows that \eqref{eq:m2:0} is an equality. 

\item Otherwise for every $j \in J$ we can take the unique coupling $(o_{j1}, \tuple{s_{j1},s_{j2}})  \in \Couplings{M_2}\big((o_{j1},s_{j1}),(o_{j2},s_{j2})\big)$ which is necessarily optimal. Continuing \eqref{eq:m2:1} we have
\begin{align}
	\Wasserstein{\PowersetFinite}{\left(\Wasserstein{M_2}{d}\right)}(T_1,T_2) &= \max_{j \in J} \EvaluationFunctor{M_2}d\left(o_{j1}, \tuple{s_{j1},s_{j2}}\right) \nonumber\\
	&= \max_{j \in J} \ev_{M_2}\Big(\big(\id_2 \times d^A\big)\big(o_{j1}, \tuple{s_{j1},s_{j2}}\big)\Big)\nonumber\\
	&= \max_{j \in J} \ev_{M_2}\Big(o_{j1}, \lambda a.d\big(s_{j1}(a),s_{j2}(a)\big)\Big)\nonumber\\
	&= c \cdot \max_{j \in J} \ev_{I} \Big(\lambda a.d\big(s_{j1}(a),s_{j2}(a)\big)\Big)\label{eq:m2:2}
	\end{align}

Then we define 
\begin{align*}
	T := \set{(o_{j1}, \tuple{s_{j1},s_{j2}}) \mid j \in J} \subseteq \PowersetFinite M_2(X \times X) = \PowersetFinite\Big(2 \times (X \times X)^A\Big)\,.
\end{align*}
We calculate for $\pi_i \colon X \times X \to X$
\begin{align*}
\PowersetFinite M_2 \pi_i (T) = (\id_2 \times \pi_i^A)[T] = \set{(o_{j1}, s_{ji}) \mid j \in J} = T_i
\end{align*}
and thus $T \in \Couplings{\PowersetFinite M_2}(T_1,T_2)$. Moreover we have
\begin{align}
\Wasserstein{\PowersetFinite M_2}{d}(T_1,T_2) &\leq \EvaluationFunctor{\PowersetFinite M_2}d(T) = \ev_{\PowersetFinite} \circ \PowersetFinite \ev_{M2} \circ \PowersetFinite M_2 d(T) \nonumber\\
&=\max\left(\PowersetFinite \left(\ev_{M2} \circ M_2 d\right)(T)\right) = \max \left((\ev_{M_2} \circ M_2 d) [T]  \right) \nonumber\\
&= \max \Big( \ev_{M2} \big[(\id_2 \times d^A)[T]\big]\Big)\nonumber \\
&=\max_{j \in J}  \ev_{M2} \Big(o_{j1}, \lambda a.d\big(s_{j1}(a), s_{j2}(a)\big)\Big)\nonumber\\
&=c \cdot \max_{j \in J}  \ev_{I} \Big(\lambda a.d\big(s_{j1}(a), s_{j2}(a)\big)\Big)
 \label{eq:m2:3}
 \end{align}
thus using this, \eqref{eq:m2:2} and \eqref{eq:m2:0} we conclude that 
\begin{align*}
\Wasserstein{\PowersetFinite M_2}{d}(T_1,T_2) &\leq c \cdot \max_{j \in J}  \ev_{I} \Big(\lambda a.d\big(s_{j1}(a), s_{j2}(a)\big)\Big)\\
&= \Wasserstein{\PowersetFinite}{\left(\Wasserstein{M_2}{d}\right)}(T_1,T_2) \leq \Wasserstein{\PowersetFinite M_2}{d}(T_1,T_2)
\end{align*}
which proves equality. \qedhere
\end{itemize}
\end{proof}

\subsection{\nameref{sec:monadlifting}}
\restate{prop:nt-lifting}
\begin{proof}
Let $t_1,t_2 \in FX$.
\begin{enumerate}
\item By naturality of $\lambda$ and $\ev_G \circ \lambda_{[0,\top]} \leq \ev_F$ we obtain for every $f \colon X \to [0,\top]$:\begin{align}
	\EvaluationFunctor{G}f \circ \lambda_X = \ev_G \circ Gf \circ \lambda_X = \ev_G \circ \lambda_{[0,\top]} \circ Ff \leq \ev_F \circ Ff = \EvaluationFunctor{F}f\,.\label{eq:nt-lifting:1}
\end{align}
Using this we compute
\begin{align}
	&\Kantorovich{G}{d}\left(\lambda_X(t_1), \lambda_X(t_2)\right) = \sup\set{d_e\left(\EvaluationFunctor{G}f\big(\lambda_X(t_1)\big),\EvaluationFunctor{G}f\big(\lambda_X(t_2)\big)\right)\, \middle|\, f \colon (X,d) \nonexpansiveTo ([0,\top],d_e)}\nonumber\\
	&\quad\leq\sup\set{d_e\left(\EvaluationFunctor{F}f(t_1),\EvaluationFunctor{F}f(t_2)\right)\, \middle|\, f \colon (X,d) \nonexpansiveTo ([0,\top],d_e)} = \Kantorovich{F}{d}(t_1,t_2)\,.\label{eq:nt-lifting:2}
\end{align}
\item We just have to replace the inequality by equality in \eqref{eq:nt-lifting:1} and \eqref{eq:nt-lifting:2}.
\item Naturality of $\lambda$ yields the following equations, where $\pi_i \colon X \times X \to X$ are the projections of the product and $d\colon X \times X \to [0,\top]$ a pseudometric on $X$.
\begin{align}
	\lambda_X \circ F\pi_i &= G\pi_i \circ \lambda_{X \times X}\label{eq:nt-lifting:3}\\
	\lambda_{[0,\top]} \circ Fd &= Gd \circ \lambda_{X \times X}\label{eq:nt-lifting:4}
\end{align}
Using \eqref{eq:nt-lifting:3}, we can see that $\lambda_{X \times X}$ maps every coupling $t$ of $t_1$ and $t_2$ to a coupling $\lambda_{X \times X}(t)$ of $\lambda_X(t_1)$ and $\lambda_X(t_2)$ because $G\pi_i(\lambda_{X \times X}(t)) = \lambda_X(F\pi_i(t)) = \lambda_X(t_i)$. Moreover, by our requirement we obtain
\begin{align*}
	\EvaluationFunctor{G}d(\lambda_{X \times X}(t)) = \ev_{G} \circ Gd \circ \lambda_{X \times X}(t) &= \ev_{G} \circ \lambda_{[0,\top]} \circ Fd (t)
\leq \ev_{F} \circ Fd (t) = \EvaluationFunctor{F}d(t)
\end{align*}
With these preparations at hand we can finally see that
\begin{align*}
	\Wasserstein{G}{d}(\lambda_X(t_1), \lambda_X(t_2)) &= \inf\set{\EvaluationFunctor{G}d(t') \mid t' \in \Couplings{G}\big(\lambda_X(t_1),\lambda_X(t_2)\big)}\\
	&\leq \inf\set{\EvaluationFunctor{G}d(\lambda_{X \times X}(t)) \mid t \in \Couplings{F}(t_1,t_2)}\\
	&\leq \inf\set{\EvaluationFunctor{F}d(t) \mid t \in \Couplings{F}(t_1,t_2)} = \Wasserstein{F}{d}(t_1,t_2)\,.
\end{align*}

\item Using the previous two results and the fact that Wasserstein is an upper bound yields: 
\begin{align*}
	 \Kantorovich{F}{d} = \Kantorovich{G}{d} \circ (\lambda_X \times \lambda_X) \leq \Wasserstein{G}{d} \circ (\lambda_X\times\lambda_X) \leq \Wasserstein{F}{d}
\end{align*}
and since $\Kantorovich{F}{d} = \Wasserstein{F}{d}$ all these inequalities are equalities.\qedhere
\end{enumerate}
\end{proof}

\restate{cor:monad-lifting}
\begin{proof}
This is an immediate consequence of \autoref{prop:nt-lifting}. For the unit take $F=\Id$ with evaluation function $\ev_F = \id_{[0,\top]}$, hence $\Kantorovich{F}{d} = \Wasserstein{F}{d}=d$ and $G=T$, $\ev_G=\ev_T$, $\lambda = \eta \colon \Id \Rightarrow T$. For the multiplication take $F = TT$, $G=T$, $\ev_F = \ev_{TT} = \ev_T \circ T\ev_T$, $\ev_G = \ev_T$ and $\lambda = \mu$.
\end{proof}

\subsection{\nameref{sec:tracemetrics}}
\restate{cor:lifting-distr-law}
\begin{proof}
For $FG$ we take the evaluation function $\ev_{FG} = \ev_F \circ F\ev_G$ and for $GF$ the evaluation function $\ev_{GF} = \ev_G \circ G\ev_F$. We have 
\begin{align*}
	\ev_{GF} \circ  \lambda_{[0,\top]} = \ev_G \circ G\ev_F \circ \lambda_{[0,\top]} \leq \ev_F \circ F\ev_G = \ev_{FG}
\end{align*}
By \autoref{prop:nt-lifting} we know that $\Wasserstein{GF}{d} \circ (\lambda_X \times \lambda_X) \leq \Wasserstein{FG}{d}$ and by \autoref{lem:compositionality2}.\ref{item:geqWasserstein} we have $\Wasserstein{G}{(\Wasserstein{F}{d})} \leq \Wasserstein{GF}{d}$. Plugging everything together we see that
\begin{align*}
	\Wasserstein{G}{(\Wasserstein{F}{d})}  \circ (\lambda_X \times \lambda_X) \leq \Wasserstein{GF}{d} \circ (\lambda_X \times \lambda_X) \leq \Wasserstein{FG}{d} = \Wasserstein{F}{(\Wasserstein{G}{d})}
\end{align*}
which is the desired nonexpansiveness.
\end{proof}

\restatewithstyle{exa:EMlaw:powfin}{restatedefinition}
\begin{proof}
The functors are a composition of known endofunctors. We have $F = \PowersetFinite(2 \times \_^A) = \PowersetFinite M_2$, and $G = 2 \times \PowersetFinite^A = M_2 \PowersetFinite$ where $M_2 := M(2,\_)$ is the endofunctor obtained from the machine bifunctor $M$ by fixing its first component to $2$. The evaluation functions are $\ev_F\colon \PowersetFinite(2 \times [0,1]^A) \to[0,1]$ where for $S \in \PowersetFinite(2 \times [0,1]^A)$ 
\begin{align*}
	\ev_F(S) &= \ev_{\PowersetFinite} \circ \PowersetFinite\ev_{M_2} (S)= \max\set{\ev_{M_2}(o,s) \mid (o,s) \in S} \\
	&=  \max\set{c \cdot \max_{a \in A} s(a) \mid (o,s) \in S} = c \cdot \max_{(o,s) \in S}\max_{a \in A}s(a)
\end{align*}
and $\ev_G \colon 2 \times (\PowersetFinite [0,1])^A \to [0,1]$ where for $(o,s) \in 2 \times (\PowersetFinite X)^A$
\begin{align*}
	\ev_G(o,s) = \ev_{M_2} \circ M_2(\ev_{\PowersetFinite})(o,s) &= \ev_{M_2}\big(o, \lambda a.\max s(a)\big) = c \cdot \max_{a \in A} \max s(a)\,.
\end{align*}

\noindent As we have seen in \autoref{exa:m2} we have compositionality for $F$. We want to apply \autoref{prop:nt-lifting}.\ref{prop:nt-lifting:3} to show nonexpansiveness. For this we have to check that the inequality $\ev_{G} \circ \lambda_{[0,1]} \leq \ev_{F}$ holds. Indeed we have:
\begin{align*}
	\ev_G \circ \lambda_{[0,1]}(S) = c \cdot \max_{a \in A}\max \set{s(a) \mid (o,s) \in S} = c \cdot \max_{a \in A}\max_{(o,s) \in S} s(a) = \ev_F(S) 
\end{align*}
which concludes the proof.
\end{proof}

\restatewithstyle{exa:EMlaw:distr}{restatedefinition}
\begin{proof}
We first quickly check that the definition is sound, i.e. that we get a probability distribution for each $a \in A$:
\begin{align*}
\sum_{x \in X}\left( \sum_{{s \in X^A, s(a)=x}}P([0,1],s)\right) &= \sum_{s \in X^A}P([0,1],s) = 1
\end{align*}

\noindent Having verified this, we now want to show nonexpansiveness. The involved bifunctors $F,G$ are given by the assignments $F(B,X) = \Distributions(B \times X^A)$ and $G(B,X) = B \times (\Distributions X)^A$ and arise from composition of the distribution functor, the identity functor and the machine bifunctor: We have $F = \Distributions \circ M$ and $G = M \circ (\Id \times \Distributions)$.

Since all of these functors have optimal couplings, we have compositionality and the canonical evaluation functions for the composed functors are
\begin{align*}
	&\ev_F := \ev_{\Distributions} \circ \Distributions\ev_M\colon\Distributions([0,1]\times [0,1]^A) \to [0,1] \quad \text{and}\\
	&\ev_G := \ev_M \circ M(\id_{[0,1]}, \ev_\Distributions) \colon [0,1] \times (\Distributions X)^A \to [0,1]\,.
\end{align*}
We will now define a function $\Lambda_X \colon \Distributions([0,1]^2\times (X\times X)^A)\to [0,1]^2\times (\Distributions (X\times X))^A$ which transfers $F$-couplings to suitable $G$-couplings in the following sense. For any $P_1,P_2 \in \Distributions([0,1] \times X^A)$ and any $P \in \Couplings{F}(P_1,P_2) \subseteq \Distributions([0,1] \times [0,1] \times (X\times X)^A)$ the function $\Lambda_X$ has to satisfy the following two requirements
\begin{align}
	&\Lambda_X(P) \in \Couplings{G}\big(\lambda_X(P_1), \lambda_X(P_2)\big)\label{eq:nt-bifunctor:cpl}\\
	& \EvaluationFunctor{G}(d_B,d)\big(\Lambda_X(P)\big)\leq \EvaluationFunctor{F}(d_B,d)(P)\label{eq:nt-bifunctor:claim2}
\end{align}
\noindent because then we have
\begin{align*}
	\Wasserstein{G}{(d_B,d)}(\lambda_X(P_1), \lambda_X(P_2)) &= \inf\set{\EvaluationFunctor{G}(d_B,d)(P') \mid P' \in \Couplings{G}\big(\lambda_X(P_1),\lambda_X(P_2)\big)}\\
	&\leq \inf\set{\EvaluationFunctor{G}(d_B,d)\big(\Lambda_{X}(P)\big) \mid P \in \Couplings{F}(P_1,P_2)}\\
	&\leq \inf\set{\EvaluationFunctor{F}(d_B,d)(P) \mid t \in \Couplings{F}(P_1,P_2)} = \Wasserstein{F}{(d_B,d)}(P_1,P_2) 
\end{align*}
\noindent which, due to compositionality, proves the desired nonexpansiveness of $\lambda_X$. So let us now define $\Lambda_X$ and prove that it satisfies the above requirements: For any set $X$ and any $P \in \Distributions([0,1] \times [0,1] \times (X\times X)^A)$ we define $\Lambda_X(P) = (\mathfrak{o}_1(P), \mathfrak{o}_2(P), \mathfrak{s}(P))$ where 
\begin{align*}
&\mathfrak{o}_1(P) = \sum_{r \in [0,1]} r \cdot P(r,[0,1],(X\times X)^A), \quad \mathfrak{o}_2(P) = \sum_{r \in [0,1]} r \cdot P([0,1],r,(X\times X)^A) \quad \text{and}\\
&\mathfrak{s}(P) \colon A \to \Distributions (X \times X), \quad \mathfrak{s}(P)(a)(x,y) = \!\!\!\!\sum_{{s \in (X \times X)^A,\ s(a)=(x,y)}}\!\!\!\!\!\!\!\!\!\!\!P([0,1]^2,s)\,.
\end{align*}

\noindent Observe that this is completely analogous to the definition of the components $\lambda_X$ of our distributive law where for any $Q \in \Distributions([0,1] \times X^A)$ we have $\lambda_X(Q) = (\mathfrak{o}'(Q), \mathfrak{s}'(Q))$ with 
\begin{align*}
	\mathfrak{o}'(Q) = \sum_{r \in [0,1]}\!r \cdot Q(r,X^A),  \quad \text{and} \quad \mathfrak{s}'(Q) \colon A \to \Distributions X,\ \ \mathfrak{s}'(Q)(a)(x) = \!\!\!\!\sum_{{s \in X^A, s(a)=x}}\!\!\!\!\!\!\!Q([0,1],s)\,.
\end{align*}

\noindent Let us now show that the above definition of $\Lambda_X$ satisfies our requirements. We thus assume from here on that $P \in \Couplings{F}(P_1, P_2)$ for some arbitrary $P_1,P_2 \in \Distributions([0,1] \times X^A)$ i.e. we know $F(\pi_i, \pi_i) = P_i$. In order to show \eqref{eq:nt-bifunctor:cpl}, we have to prove that the equation
\begin{align}
	G(\pi_i, \pi_i) \big(\Lambda_X(P)\big) = \lambda_X(P_i) \label{eq:nt-bifunctor:2}
\end{align}
holds. The left hand side of this equation evaluates to
\begin{align}
G(\pi_i, \pi_i) \big( \Lambda_X (P) \big) &= \big(\pi_i \times (\Distributions\pi_i)^A\big) \left(\mathfrak{o}_1(P), \mathfrak{o}_2(P), \mathfrak{s}(P)\right) = \big(\mathfrak{o}_i(P), \Distributions\pi_i \circ \mathfrak{s}(P)\big)\label{eq:nt-bifunctor:2:lhs}
\end{align}
and since $F(\pi_i, \pi_i) = P_i$ the right hand side of \eqref{eq:nt-bifunctor:2} evaluates to 
\begin{align}
\lambda_X (P_i) = \lambda_X \big( F(\pi_i, \pi_i) (P) \big) &= \Big(\mathfrak{o}'\left(\Distributions(\pi_i \times \pi_i^A)(P)\right),\mathfrak{s}'\left(\Distributions(\pi_i \times \pi_i^A)(P)\right)\Big)\,.\label{eq:nt-bifunctor:2:rhs}
\end{align}

\noindent In order to prove \eqref{eq:nt-bifunctor:2} we will thus have to show that $\mathfrak{o}'\left(\Distributions(\pi_i \times \pi_i^A)(P)\right) = \mathfrak{o}_i(P)$ and also $\mathfrak{s}'\left(\Distributions(\pi_i \times \pi_i^A)(P)\right) = \Distributions\pi_i \circ \mathfrak{s}(P)$ holds. We first compute
\begin{align*}
	\quad & \mathfrak{o}'\left(\Distributions(\pi_i \times \pi_i^A)(P)\right) = \sum_{r \in [0,1]}r \cdot \Distributions(\pi_i \times \pi_i^A)(P)(r,X^A)\\
	&=\sum_{r \in [0,1]}r \cdot \left(P \circ (\pi_i \times \pi_i^A)^{-1}[\set{r} \times X^A]\right)\\
	&=\sum_{r \in [0,1]}r \cdot P \left(\set{(o_1,o_2,s) \in \Distributions([0,1]^2 \times (X\times X)^A) \ \middle|\ \pi_i \times \pi_i^A(o_1,o_2,s) \in \set{r} \times X^A}\right)\\
	&=\sum_{r \in [0,1]}r \cdot P \left(\set{(o_1,o_2,s) \in \Distributions([0,1]^2 \times (X\times X)^A) \ \middle|\ o_i = r}\right) = \mathfrak{o}_i(P)
\end{align*}
showing that indeed the first components of the tuples in \eqref{eq:nt-bifunctor:2:lhs} and \eqref{eq:nt-bifunctor:2:rhs} are the same. For the second components we have
\begin{align*}
	\quad &\mathfrak{s}'\left(\Distributions(\pi_i \times \pi_i^A)(P)\right)(a)(x) = \sum_{\set{s \in X^A, s(a)=x}}\Distributions(\pi_i \times \pi_i^A)(P)([0,1],s)\\
	&=\sum_{{s \in X^A, s(a)=x}}\!\!\!\!\!P\left(\set{(o_1,o_2,s') \in \Distributions([0,1]^2 \times (X^2)^A) \ \middle|\ \pi_i \times \pi_i^A(o_1,o_2,s') \in [0,1]\times \set{s}}\right)\\
	&=\sum_{{s \in X^A, s(a)=x}}\!\!\!\!\!P\left(\set{(o_1,o_2,s') \in \Distributions([0,1]^2 \times (X^2)^A) \ \middle|\ \pi_i \circ s' = s}\right)\\
	&=\sum_{{s' \in (X\times X)^A, \pi_i \circ s'(a)=x}}P([0,1]^2,s')
\end{align*}
\noindent and
\begin{align}
(\Distributions\pi_i \circ \mathfrak{s}(P))(a)(x) & = \mathfrak{s}(P)(a) \circ \pi_i^{-1}[\set{x}] = \mathfrak{s}(P)(a)\left(\set{y \in X \times X \mid \pi_i(y) = x}\right)\nonumber\\
&= \sum_{{s \in (X \times X)^A, \pi_i \circ s(a)=x}}\!\!\!\!\!\!\!P([0,1]^2,s)
\end{align}
which shows that also the second components of \eqref{eq:nt-bifunctor:2:lhs} and \eqref{eq:nt-bifunctor:2:rhs} coincide. Therefore \eqref{eq:nt-bifunctor:2} holds i.e. we have proved $\Lambda_X(P) \in \Couplings{G}(\lambda_X(P_1), \lambda_X(P_2))$ as claimed in \eqref{eq:nt-bifunctor:cpl}. We now show \eqref{eq:nt-bifunctor:claim2}. For the left hand side of that inequality we compute
\begin{align}
\EvaluationFunctor{G}(d_B,d) \big( \Lambda_X (P)\big)&= \big(\ev_G \circ G(d_B,d)\big)\big(\lambda_X(P)\big) = \ev_G \Big( G(d_B,d) \big( \Lambda_X (P) \big)\Big)\nonumber\\
&= \ev_G \Big(\big(d_B \times (\Distributions d)^A\big) \big(\mathfrak{o}_1(P), \mathfrak{o}_2(P), \mathfrak{s}(P)\big)\Big)\nonumber\\
&= \ev_G\Big(d_B \big(\mathfrak{o}_1(P),\mathfrak{o}_2(P)\big), \lambda a. \Distributions d\big(\mathfrak{s}(P)(a)\big)\Big)\nonumber\\
&=\Big(\ev_M \circ M\left(\id_{[0,1]}, \ev_\Distributions\right) \Big)\Big(d_B \big(\mathfrak{o}_1(P),\mathfrak{o}_2(P)\big), \lambda a. \Distributions d\big(\mathfrak{s}(P)(a)\big)\Big)\nonumber\\
&=\ev_M\left(M\left(\id_{[0,1]}, \ev_\Distributions\right)\Big(d_B \big(\mathfrak{o}_1(P),\mathfrak{o}_2(P)\big), \lambda a. \Distributions d\big(\mathfrak{s}(P)(a)\big)\Big)\right) \nonumber\\
&=\ev_M\Bigg(d_B \big(\mathfrak{o}_1(P),\mathfrak{o}_2(P)\big), \ev_\Distributions^A\Big(\lambda a. \Distributions d\big(\mathfrak{s}(P)(a)\big)\Big)\Bigg) \nonumber\\
&=\ev_M\Bigg(d_B \big(\mathfrak{o}_1(P),\mathfrak{o}_2(P)\big), \lambda a.\ev_\Distributions\Big(\Distributions d\big(\mathfrak{s}(P)(a)\big)\Big)\Bigg)\nonumber\\
&=c_1d_B \big(\mathfrak{o}_1(P),\mathfrak{o}_2(P)\big) + \frac{c_2}{|A|}\sum_{a \in A}\ev_\Distributions\Big(\Distributions d\big(\mathfrak{s}(P)(a)\big)\Big)\label{eq:nt-bifunctor:claim2:1}
\end{align}
and since for each $a \in A$ we have
\begin{align*}
\ev_\Distributions\Big(\Distributions d\big(\mathfrak{s}(P)(a)\big)\Big) &= \sum_{r \in [0,1]} r \cdot \mathfrak{s}(P)(a) \big( d^{-1} [\set{r}] \big) = \sum_{(x,y) \in X^2} d(x,y) \cdot \mathfrak{s}(P)(a)(x,y)\\
&= \sum_{(x,y) \in X^2} d(x,y) \cdot \left(\sum_{{s \in (X \times X)^A,\ s(a)=(x,y)}}\!\!\!\!\!\!\!\!\!\!\!P([0,1]^2,s)\right)\\
&=\sum_{s \in (X\times X)^A} d\big(s(a)\big) \cdot P([0,1]^2, s)
\end{align*}
we may continue \eqref{eq:nt-bifunctor:claim2:1} as follows:
\begin{align}
\EvaluationFunctor{G}(d_B,d) \big( \Lambda_X (P)\big)&= c_1d_B \big(\mathfrak{o}_1(P),\mathfrak{o}_2(P)\big) + \frac{c_2}{|A|}\sum_{a \in A}\sum_{s \in (X\times X)^A} d\big(s(a)\big) \cdot P([0,1]^2, s)\,.\label{eq:nt-bifunctor:claim2:2}
\end{align}
For the right hand side of \eqref{eq:nt-bifunctor:claim2} we have
\begin{align}
	\EvaluationFunctor{F}&(d_B,d)(P) = \big(\ev_F \circ {F}(d_B,d)\big)(P) = \Big(\big(\ev_{\Distributions} \circ \Distributions\ev_M\big) \circ \Distributions\big(M(d_B, d)\big)\Big)(P)\nonumber\\
	&=\ev_{\Distributions} \Bigg( \Distributions\ev_M \Big(\Distributions\big(M(d_B, d)\big)(P)\Big) \Bigg) = \ev_{\Distributions} \Bigg( \Distributions\ev_M \Big(\Distributions\big(d_B \times d^A\big)(P)\Big) \Bigg)\nonumber\\
	&= \ev_{\Distributions} \Bigg(\Big(\Distributions\big(d_B \times d^A\big)(P)\Big) \circ \ev_M^{-1} \Bigg)=\sum_{r \in [0,1]} r \cdot \Big(\Distributions\big(d_B \times d^A\big)(P)\Big)\big(\ev_M^{-1}[\set{r}]\big)\nonumber\\
	&=\sum_{(o,s') \in [0,1] \times [0,1]^A} \!\!\!\! \ev_M(o,s') \cdot \Big(\Distributions\big(d_B \times d^A\big)(P)\Big)(o,s')\nonumber\\
	&=\sum_{(o,s') \in [0,1] \times [0,1]^A} \!\!\!\!\ev_M(o,s') \cdot P\Big(\big(d_B \times d^A\big)^{-1}\big[\set{(o,s')}\big]\Big)\nonumber\\
	&= \sum_{(o_1,o_2,s) \in [0,1]^2 \times (X\times X)^A} \ev_M\Big(\big(d_B \times d^A\big)(o_1,o_2,s)\Big) \cdot P(o_1,o_2,s)\nonumber\\
	&= \sum_{(o_1,o_2,s) \in [0,1]^2 \times (X\times X)^A} \ev_M\Big(d_B(o_1,o_2), \lambda a.d\big(s(a)\big)\Big) \cdot P(o_1,o_2,s)\nonumber\\
	&= \sum_{(o_1,o_2,s) \in [0,1]^2 \times (X\times X)^A}\left( c_1d_B(o_1,o_2) + \frac{c_2}{|A|}\sum_{a\in A} d\big(s(a)\big)\right) \cdot P(o_1,o_2,s)\nonumber\\
	&= c_1 \mathfrak{o}(P) + \frac{c_2}{|A|}\sum_{s \in (X \times X)^A}\sum_{a\in A} d\big(s(a)\big) \cdot P([0,1]^2,s)\label{eq:nt-bifunctor:claim2:3}
\end{align}
with
\begin{align*}
\mathfrak{o}(P) = \sum_{(o_1,o_2) \in [0,1]^2}\!\! d_B(o_1,o_2) \cdot P\big(o_1,o_2,(X\times X)^A\big)\,.
\end{align*}
Comparing \eqref{eq:nt-bifunctor:claim2:2} and \eqref{eq:nt-bifunctor:claim2:3} we see that in order to obtain inequality \eqref{eq:nt-bifunctor:claim2} we just have to show
\begin{align*}
d_B\big(\mathfrak{o}_1(P),\mathfrak{o}_2(P)\big)\leq\mathfrak{o}(P) 
\end{align*}
This is easily done using the fact that $d_B = d_e$ is the euclidean metric and the triangle inequality:
\begin{align*}
	d_B\big(\mathfrak{o}_1(P),&\mathfrak{o}_2(P)\big) = \\
	&= \left|\sum_{r_1 \in [0,1]} r_1 \cdot P\big(r_1,[0,1],(X\times X)^A\big) - \sum_{r_2 \in [0,1]} r_2 \cdot P\big([0,1],r_2,(X\times X)^A\big)\right|\\
	&=\left|\sum_{r_1,r_2 \in [0,1]} r_1 \cdot P\big(r_1,r_2,(X\times X)^A\big) - \sum_{r_1,r_2 \in [0,1]} r_2 \cdot P\big(r_1,r_2,(X\times X)^A\big)\right|\\
	&=\left|\sum_{r_1,r_2 \in [0,1]} (r_1 - r_2) \cdot P\big(r_1,r_2,(X\times X)^A\big)\right|\\
	&\leq \sum_{r_1,r_2 \in [0,1]} \left|(r_1 - r_2) \cdot P\big(r_1,r_2,(X\times X)^A\big)\right| \\
	&= \sum_{r_1,r_2 \in [0,1]} |r_1 - r_2| \cdot P\big(r_1,r_2,(X\times X)^A\big) = \mathfrak{o}(P)\,.
\end{align*}
We have thus also completed the proof of our second claim: the inequality \eqref{eq:nt-bifunctor:claim2}. 
\end{proof} 
\end{document}